\documentclass[journal,twoside,web]{ieeecolor}
\usepackage{generic}
\usepackage{cite}
\usepackage{amsmath,amssymb,amsfonts}
\usepackage{graphicx}
\usepackage{algorithm}
\usepackage{textcomp}
\usepackage{algorithm}
\usepackage{mathtools,verbatim}
\usepackage{epstopdf}

\usepackage{empheq}
\usepackage{algpseudocode}
\usepackage{mathbbol}
\usepackage{dirtytalk}              
\usepackage{amsmath}

\usepackage{verbatim}
\usepackage{subcaption}

\usepackage{amssymb}
\newtheorem{thm}{\textbf{Theorem}}
\newtheorem{assump}{\textbf{Assumption:}}
\newtheorem{lemma}{\textbf{Lemma}}

\newtheorem{definition}{\textbf{Definition}}
\newtheorem{remark}{\textbf{Remark}}
\newtheorem{expm}{Example}

\usepackage{subcaption}
\def\BibTeX{{\rm B\kern-.05em{\sc i\kern-.025em b}\kern-.08em
    T\kern-.1667em\lower.7ex\hbox{E}\kern-.125emX}}
\begin{document}
\title{Influence Enhancement in Opinion Dynamics {U}sing Edge Modification: \\
A Kron Reduction-Based Approach}

\author{Aashi Shrinate, \IEEEmembership{Student, IEEE}, 
Aravind Seshadri, Twinkle Tripathy, \IEEEmembership{Senior Member, IEEE},
Laxmidhar Behera, \IEEEmembership{Senior Member, IEEE}, 
Lingfei Wang, Karl Henrik Johansson \IEEEmembership{Fellow, IEEE}
\thanks{This work was supported in part by the Prime Minister's Research Fellowship and by the Knut and Alice Wallenberg Foundation Wallenberg Scholar Grant, the Swedish Research Council Distinguished Professor Grant 2017-01078, and the Swedish Foundation for Strategic Research SUCCESS FUS21-0026.}
\thanks{
$^{1}$Aashi Shrinate, 
Twinkle Tripathy,  Laxmidhar Behera are with the Department of Electrical Engineering, Indian Institute of Technology Kanpur. \\
 $^2$ Aravind Seshadri is with Adobe Systems, Bangalore. \\
 $^3$ Lingfei Wang and Karl Henrik Johansson are with School of Electrical Engineering and Computer Science, KTH Royal Institute of Technology, Stockholm.
 \textit{Email:} \texttt{ aashis21@iitk.ac.in , aravi15062003@gmail.com, ttripathy@iitk.ac.in, lbehera@iitk.ac.in,lingfei@kth.se,kallej@kth.se.}}}

\maketitle

\begin{abstract}
With the emergence of online social networks as a primary platform for advertising and advocacy, enhancing a user's influence has become of significant interest. In this paper, we investigate this problem under the Friedkin-Johnsen {opinion dynamics} model, wherein stubborn agents influence the opinions of other agents in the network. Unlike most of the existing works, we leverage topological properties of the network to {increase} the influence of a desired stubborn agent. 
Specifically, we introduce the notion of edge modification{, which} mimics the mechanism of recommendations  in {social networks}. 
 First, we present a topology-based condition that identifies edge modifications that always increase the influence of a desired stubborn agent. 
{It is shown that} the impact of the chosen edge modifications remains robust
 to changes in parameters such as stubbornness and the interaction weights. Thereafter, we formulate a discrete optimisation problem to identify a set of edge modifications that maximise the agent's influence centrality. 
 We present a computationally efficient approximate solution to the optimisation problem. Finally, we demonstrate the effectiveness of our {approach on the Friedkin-Johnsen opinion dynamics over} the Erdős–Rényi random graph.   
\end{abstract}


\section{Introduction}
Certain individuals in a social network influence the outcome of the discussions within the network. This ability of an individual to shape the collective outcome of networked interactions is defined as their social power \cite{cartwright1959studies}. An accurate estimate of social power {can be} derived from the underlying dynamics of opinion evolution in a network \cite{friedkin1991centrality}.
Owing to conforming empirical evidence, the DeGroot's averaging-based mechanism is widely adopted \cite{degroot1974consensus,anderson1981foundations}. The Friedkin-Johnsen (FJ) model further extends this process by incorporating the stubborn behaviour of {some} agents (their attachment to biases) to explain the common occurrence of persistent disagreement. In the FJ framework, the influence centrality (social power) of an agent quantifies the impact of its bias on the average final opinion \cite{Community_Cleavage}. Henceforth, we refer to an agent's social power as its influence centrality. Recently, the problem of enhancing a particular agent's influence centrality has attracted significant interest due to its wide applications, from political campaigns to brand marketing. 
\vspace{-10pt}
\subsection{Motivation}
A particular agent's influence centrality depends on the stubbornness of {various} agents and {the interaction network}.
Existing approaches \cite{10.1145/3511808.3557304,Lingfei_wang,L_wang_parallel} {have shown how to} suitably allocate an agent's stubbornness to enhance its influence centrality. However, in the FJ framework, stubbornness is generally derived from an agent's inherent features, which \textit{might be} infeasible to modify. Specifically, in \cite{Lingfei_wang,L_wang_parallel}, the authors equate stubbornness with the trait of outspokenness, which may vary across distinct issues. Nevertheless, unless the interaction network is conducive, being highly vocal (stubborn) does not guarantee a high influence centrality. Moreover, the impact of network topology on influence centrality is evident in online social networks, where the popularity of influencers is shaped by the algorithms that recommend their content \cite{feed_algo}. Consequently, suitable network interventions emerge as an effective and practically feasible tool for increasing an agent's influence centrality. However, designing such interventions remains challenging due to the complex, nonlinear dependence of influence centrality on the network topology \cite{Community_Cleavage}. 



Motivated by this, the present work proposes a tactical addition of edges in the network, akin to
algorithmic recommendations in online social networks, for enhancing the influence centrality of an agent. 
In recent works  \cite{ancona2022model,wang2025addinglinks}, edges are added from the desired stubborn agent to increase its influence centrality. However, these works rely on restrictive assumptions on network connectivity and the {desired} agent's dynamics.
On the contrary, our work seeks a criterion for identifying the suitable edges that increase the influence centrality of any selected stubborn agent in an arbitrary directed graph. 
Importantly, the proposed framework results in a significantly broader range of viable edges.
\vspace{-10pt}
\subsection{Related literature}

In recent works \cite{ancona2022model,gt_attract,wang2025addinglinks}, the authors assume that an external stubborn agent exists in the network whose opinion remains fixed, unaffected by any other agent under the opinion evolution process. These works increase the external agent's influence centrality by effectively designing its interactions with the remaining agents in the network. The authors in  \cite{ancona2022model} present an FJ model-based pro-vaccine campaign in which edges are added from the external influencer to selected target nodes. 
Reference \cite{gt_attract} examines a zero-sum game between two external stubborn agents, wherein each agent suitably allocates the weights to its interactions in the network to maximise its influence centrality. Further, under the influence maximisation problem examined in \cite{wang2025addinglinks}, an external stubborn agent adds a fixed number of edges from itself to other agents to maximise its influence centrality. Along similar lines,  the related problem of opinion maximisation is solved in \cite{ZHU2025115090}, where an external leader drives the average final opinion close to its own opinion by adding edges from itself. 



In the context of social networks, the approach of adding edges to increase influence centrality is especially pertinent as sponsored content is promoted through targeted recommendations \cite{feed_algo}. Standard recommender algorithms, including content-based \cite{content_filtering}, collaborative \cite{collab_filtering}, and hybrid filtering \cite{Survey_recommendations}, identify the target audience by exploiting users' prior preferences and/or the similar interests of like-minded users. However, these mechanisms largely overlook the social interactions among users  \cite{Survey_recommendations}. In practice, social interactions can create a cascading effect that amplifies an entity’s popularity across the social network. The authors in \cite{ancona2022model} design targeted vaccine campaigns that incorporate the effect of social interactions via the FJ model. Such mechanisms are especially effective in designing campaigns that aim to influence the entire population, such as awareness campaigns, countering fake news, and even brand marketing. 

{The existing approaches for adding edges (designing recommendations) under the FJ model presented in \cite{wang2025addinglinks,gt_attract,ancona2022model} rely on the assumption that the external agent is unaffected by the opinions of other agents.} Without this assumption, a strict increase in influence centrality through edge additions from the desired stubborn agent is no longer guaranteed (see Example \ref{expm:null_mod}).  Additionally, the underlying network, excluding the external agent, is assumed to be an undirected connected graph in \cite{ancona2022model} and a strongly connected directed graph in \cite{gt_attract,wang2025addinglinks}. 
These assumptions are restrictive because (i) {social networks} can take the structure of any weakly connected digraph and (ii) the influencers and their followers often engage in two-way communication through likes, dislikes and comments. As a result, such influencers often adjust their views in response to the audience feedback \cite{etienne2024mimetic}. 

\vspace{-10pt}
\subsection{Contributions}
In this work, we consider a directed graph (digraph) comprising $n$ individuals whose opinions evolve according to the FJ model. Each agent's stubbornness is considered to be its intrinsic property that remains fixed. We seek to increase a {given} stubborn agent's influence centrality by suitably altering the network interactions. To this end, we use the mechanism of edge modification $(a,b,d)$, {which} mimics algorithmic recommendations in social networks. Under {a modification} $(a,b,d)$, an edge $(a,b)$ is added to the network and the edge-weight of an existing edge $(d,b)$ is reduced such that the in-degree of $b$ remains constant. Unlike the edge addition protocol in \cite{wang2025addinglinks}, which reduces the weights of all incoming edges at $b$, the proposed approach enables the reduction in the weight of selected edge. 
This mechanism allows a recommender algorithm to maintain thematic consistency in recommendations by specifically deprioritising unrelated content rather than penalising all sources.

The primary objective of this work is to characterise edge modifications that enhance the influence centrality of any stubborn agent in any arbitrary digraph, thereby generalising the results in \cite{wang2025addinglinks,ancona2022model}. 
We begin by establishing that a stubborn agent's influence centrality increases under edge modification $(a,b,d)$, if and only if the aggregate weight of all the walks from the stubborn agent to $a$ exceeds that of the walks to $d$ (Theorem \ref{Lemma:path-weights}). However, {verification of} this condition is both computationally expensive and depends on the parameters of the FJ model, including stubbornness and edge weights, which are hard to accurately estimate {in practice}. To overcome these challenges, we employ the {notion} of locally topologically persuasive (LTP) agent introduced in \cite{shrinate2025opinionclusteringfriedkinjohnsenmodel}. Using this notion, for each stubborn agent, we classify the agents in the network as its endorsers (the stubborn LTP agent and those persuaded by it) and its non-endorsers (the rest). Algorithm \ref{algo:identifying_endorsers} presents a computationally efficient way to identify these agents.
{Using Kron reduction \cite{Kron_red_digraphs}, we show that when node $a$ is the endorser of a particular stubborn agent, the impact of $(a,b,d)$ on the stubborn agent's influence centrality is independent of the FJ model parameters.} More importantly, we present a topology-based sufficient condition to increase the influence centrality of any stubborn agent. 

Our key contributions are as follows:

\begin{itemize}
\item {First, we derive an augmented graph from the FJ dynamics.}
A special {Kron reduced} version of the augmented graph always exists with the weight of each edge  equal to the associated stubborn agent’s influence centrality (Lemma \ref{lm:influence_cen}).
Consequently, the impact of {edge modification} $(a,b,d)$ is determined by examining the change in the edge weights of this reduced graph. {W}e identify the affected edges in any Kron-reduced graph when an edge is {modified} in the original graph (Theorem \ref{lemma:modification_propagates}).



\item 
    
 
 %
{Second, we define a stubborn agent's endorsers and non-endorsers and present Algorithm \ref{algo:identifying_endorsers} to identify these agents.} We show that a stubborn agent's influence centrality always increases under  {modification} $(a,b,d)$ if $a$ is {an} endorser and $d$ is {not} (Theorem \ref{thm:useful_mod}), independent of the FJ model parameters. Further, unlike \cite{wang2025addinglinks,ancona2022model},
 a stubborn agent's influence centrality can be increased by adding edges from endorsers of the stubborn agent, thereby expanding the set of suitable edge modifications.


\item 
Finally, we examine the influence maximisation problem that selects a fixed number of edge modifications $(a,b,d)$ to maximise a stubborn agent's influence centrality.
Leveraging the endorser-based condition, we present a heuristic solution to this discrete optimisation problem that significantly reduces the search space by always selecting $a$ as {the stubborn agent's} endorser and optimising only over $b$ and $d$ {(Algorithm \ref{algo:greedy_optimal_edge_m})}. Compared to the standard greedy approach {(Algorithm \ref{algo:optimal_edge_m})},
the proposed algorithm is {more} efficient. 
\end{itemize}
\subsection{Organisation of the paper}
The paper is organised as follows: Sec. \ref{sec:prelims} presents notations and preliminaries. Sec. \ref{sec:FJ_Model} {gives} the problem formulation. The analysis of the impact of an edge modification is presented in Sec. \ref{sec:influence_centrality}. The influence centrality maximisation problem is analysed in Sec. \ref{sec:algorithm} and {its performance is validated in Sec. \ref{sec:simulations}.}
Finally, Sec. \ref{sec:conclusion} concludes with future research directions.

\section{Preliminaries}
\label{sec:prelims}
\subsection{Notations} Let $\mathbb{N}$ and $\mathbb{R}$ denote the set of natural and real numbers, respectively. The set of positive (non-negative) real numbers is denoted as $\mathbb{R}_{>0}$( $\mathbb{R}_{\geq 0}$). $\mathbf{1}_n$ ($\mathbf{0}_n$) is a column vector with $n$ entries, each equal to $1$ ($0$), with dimensions omitted if no confusion arises. $I$ denotes the identity matrix of appropriate dimension. The vector $e_i$ is the standard basis vector with 
a single non-zero entry at the $i^{th}$ position that equals $1$. Given a finite set $X$, its cardinality is given by $|X|$.  The set $\{1,2,...,n\}$ for $n \in \mathbb{N}$  is denoted as $[n]$. Consider a matrix $H=[h_{ij}]$, 
its  spectral radius is denoted as $\rho(H).$ 
A matrix $H=\operatorname{diag}(h_1,h_2,...,h_n)$ is a diagonal matrix. 
Let $H\in \mathbb{R}^{p \times p}$ and $\alpha,\beta \subseteq [p]$ be the index sets. The submatrix of $H$ with rows indexed by $\alpha$ and columns indexed by $\beta$ is denoted by $H[\alpha,\beta]$. Further, $H[\alpha,\alpha]$ is simply denoted as $H[\alpha]$. Equivalently, for a vector $\mathbf{x} \in \mathbb{R}^p$, we denote a vector containing entries of $\mathbf{x}$ indexed by $\alpha$ as $x[\alpha]$.  

\subsection{Graph {theory} }Consider a group of $n$ agents whose interaction topology is represented by a digraph $\mathcal{G}=(\mathcal{V},\mathcal{E})$, where $\mathcal{V}$ is the set of nodes and $\mathcal{E} \subseteq \mathcal{V} \times \mathcal{V}$ is the set of edges. 
The adjacency matrix $W=[w_{ij}]$ of the $\mathcal{G}$ is defined as: $w_{ij}>0$ if there is an edge $(j,i) \in \mathcal{E}$, otherwise $w_{ij}=0$. The Laplacian matrix of $\mathcal{G}$ is given as $L=D-W$, where $D=\operatorname{diag}(d_1,...,d_n)$ is the in-degree matrix with $d_i=\sum_{j=1}^{n}w_{ij}$. The off-diagonal entries of $L$ are non-positive and $L$ satisfies $L\mathbf{1}_n=\mathbf{0}_n$. 

The loopy Laplacian matrix $Q$ of a digraph is defined as $Q=L+\operatorname{diag}{(w_{11},...,w_{nn})}$ \cite{dorfler2012kron}. If a digraph does not have any self-loops, then $Q=L$ is referred to as the loopless Laplacian matrix.
For any matrix $H=[h_{ij}]\in \mathbb{R}^{n \times n}$, the associated digraph 
${G}(H)$ has $n$ nodes. An edge $(j,i)$ exists in $G(H)$ with edge weight equal to  $h_{ij}$ only if $h_{ij} \neq 0$. 
 
 A \textit{walk} is an ordered sequence of nodes such that each pair of consecutive nodes forms an edge in the graph. A \textit{cycle} is a walk whose initial and final nodes coincide. If none of the nodes in a walk are repeated, it is a \textit{path}. The product of the weights of all the edges that form a path (walk) is defined as its path (walk) weight. 
A graph is said to be strongly connected if there exists a path from each node to all other nodes in the graph. 

\subsection{Matrix {theory}} A matrix $H$ of the form: $H=sI-B$, where $s>0$ and $B$ is a non-negative matrix, is an M matrix if $s\geq\rho(B)$. Each eigenvalue of an M matrix has a non-negative real part.
 If $H$ is a non-singular M matrix, it satisfies the following properties (Theorem 2.3 and 2.4 \cite{Plemmons}):
\begin{itemize}
    \item $H^{-1}=1/s\sum_{k=0}^{\infty}(B/s)^k$,
    \item all principal minors of $H$ are positive.
\end{itemize}

A real square matrix is a P matrix if all its principal minors are positive. Hence, each non-singular M matrix is also a P matrix. A matrix $H=[h_{ij}]$ is called \textit{row-diagonally dominant} if it satisfies the following: $|a_{ii}|\geq \sum_{j=1,j\neq i}^n|a_{ij}|$ for all $i\in[n]$. $H$ is called diagonally dominant of its column entries if $|a_{ii}|\geq|a_{ij}|$ for all $i\in[n]$.
\begin{lemma}(Corollary-2 \cite{JOHNSON202484})
\label{Lemma:RDD_Invertible}
    Let matrix $H\in \mathbb{R}^{n \times n}$ be a row diagonally dominant and invertible matrix. Then, $H^{-1}$ is diagonally dominant of its column entries.
\end{lemma}

\begin{lemma}
\label{lm:Q_is_M_matrix}
 A loopy (or loopless) Laplacian Matrix $Q\in \mathbb{R}^{n \times n}$ is an M matrix. 
\end{lemma}
\begin{proof}
A loopless Laplacian is an M matrix because the real part of each of its eigenvalues is non-negative  (Pg. 258 in\cite{CHEBOTAREV2002253}). 
 Since a loopy Laplacian matrix $Q$ is defined as $Q=L+\operatorname{diag}(w_{11},...,w_{nn})$, the eigenvalues of $Q$ also have a non-negative real part by the Gersgorin disks Theorem \cite{bullo}. Hence, $Q$ is also an $M$ matrix. \end{proof}

 %

\subsection{Schur complement and Kron reduction } Consider $H\in \mathbb{R}^{n \times n}$ and $\alpha \subseteq [n]$. 
Let $\alpha^{c}$ be defined as $\alpha^{c}=[n] \setminus \alpha$. If $H[\alpha^c]$ is non-singular, then the Schur complement of $H[\alpha^c]$ in $H$ is given as
\begin{align}
\label{eqn:Schur_complement}
    H/\alpha^c=H[\alpha]-H[\alpha,\alpha^c](H[\alpha^c])^{-1}H[\alpha^c,\alpha]
\end{align}

Consider a set of linear equations of the form $H\mathbf{x}=\mathbf{y}$, where 
$\mathbf{x},\mathbf{y}\in \mathbb{R}^n$. By Schur complement, we can obtain a reduced set of linear equations for $\alpha\subset [n]$, given by
\begin{align}
\label{eqn:reduced_LE}
H/\alpha ^c \cdot \mathbf{x}[\alpha ]=\mathbf{y}[\alpha ]-H[\alpha ,\alpha ^c]H[\alpha ^c]^{-1}\mathbf{y}[\alpha ^c]    
\end{align}

Let $Q\in \mathbb{R}^{n \times n}$ be a  Laplacian (or a loopy Laplacian) matrix with ${G}(Q)$ as its associated digraph. {For a suitable $\alpha\subseteq[n]$, we can evaluate the Schur complement $Q/\alpha^c$ with $G_{\alpha}$ being its associated digraph. 
The transformation of $G(Q)$ to the reduced graph $G_{\alpha}$ using the Schur complement that eliminates the nodes in set $\alpha^c$ from $G(Q)$ is called \textit{Kron reduction}.} The relation between graph-theoretic and spectral properties of the original graph and the Kron-reduced graph is presented in \cite{dorfler2012kron} for undirected graphs, and generalised to digraphs in \cite{Kron_red_digraphs}.




\begin{lemma}[\hspace{-0.5pt}\cite{Kron_red_digraphs}]
\label{lm:basic_properties_1}
Under Kron reduction of a digraph ${G}(Q)$ to $G_{\alpha}$ for $\alpha \subset[n]$, the following statements hold:
\begin{enumerate}
    \item The Kron-reduced matrix $Q/\alpha^c$ is well-defined if every node $i \in \alpha^c$ has a path in ${G}(Q)$ from a node $j \in \alpha$.\footnote{In \cite{Kron_red_digraphs}, a Kron-reduced matrix $Q/\alpha^c$ is well-defined if for each $i\in \alpha^c$ there exists a node $j\in \alpha$ and a path $i\to j$ in $G(Q)$. In this work, since the adjacency matrix follows the convention $w_{ij} > 0$ for edge $(j,i) \in \mathcal{E}$ (opposite to that in \cite{Kron_red_digraphs}), we adapt this condition accordingly.}
    \item If $Q$ is a loopless (loopy) Laplacian matrix, then the Kron-reduced matrix  $Q/\alpha^c$ is also a loopless (loopy) Laplacian matrix.
    \item If $G(Q)$ is strongly connected, then $G_{\alpha}$ is also strongly connected.
     \item An edge $(i,j)$ exists in $G_{\alpha}$ if and only if there is a path $i$ to $j$ in $G(Q)$ such that all the nodes comprising the path  {belong to $\{i,j\} \cup \alpha^c$.} 
\end{enumerate}
\end{lemma}
\begin{figure}[h]
    \centering
    \begin{subfigure}{0.2\textwidth}
    \centering
         \includegraphics[width=0.5\linewidth]{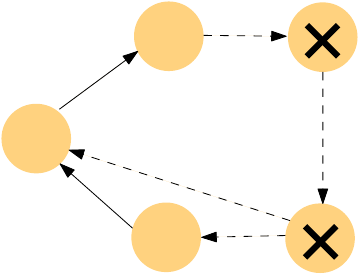}
    \caption{Network $G(Q)$}
    \label{fig:Original_Net}
    \end{subfigure}
        \begin{subfigure}{0.2\textwidth}
        \centering
         \includegraphics[width=0.3\linewidth]{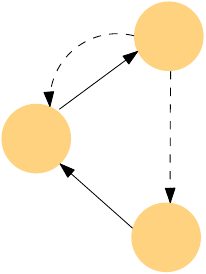}
    \caption{reduced graph $G_{\alpha}$}
    \label{fig:reduced_net}
    \end{subfigure}
    \caption{Kron reduction of $G(Q)$ to $G_{\alpha}$. Nodes marked with crosses in $G(Q)$ belong to $\alpha^c$ and the remaining are in $\alpha$. The dashed lines in {(a)} denote paths {of} $G(Q)$ that start and end at nodes in set $\alpha$ but pass through only nodes in $\alpha^c$. {As established in Statement 4) of Lemma \ref{lm:basic_properties_1}, such paths result in the dashed edges in $G_{\alpha}$ as shown in (b)}.}
    \label{fig:Example_1}
\end{figure}
\vspace{-15pt}
\section{Problem Formulation}
\label{sec:FJ_Model}

In this section, we begin with a brief review of the FJ model. Subsequently, we define the edge modification and show its impact on influence centrality. Finally, we state the problem of enhancing a stubborn agent's influence centrality using edge modifications.
\subsection{FJ model}
Consider a network $\mathcal{G}=(\mathcal{V},\mathcal{E})$ composed of $n$ agents. {The agents' opinions evolve under the FJ model as} 
\begin{equation}
\label{eq:opinion}
\mathbf{x}(k+1)=(I_n-\beta)W\mathbf{x}(k)+\beta \mathbf{x}(0)
\end{equation}
where $\mathbf{x}(k)=[x_1(k),....,x_n(k)]\in \mathbb{R}^n $ denotes the opinions, $W$ is a row-stochastic matrix and  $\beta={\operatorname{diag}}(\beta_1,...,\beta_n)$ is a diagonal matrix with $\beta_i\in[0,1]$ quantifying the stubbornness of agent $i\in\mathcal{V}$. 
An agent $i$ is stubborn if  $\beta_i>0$, and non-stubborn otherwise.  
We {assume} digraph $\mathcal{G}$ to be weakly connected and contain $m>1$ stubborn agents. Additionally, to ensure that each agent
 is affected by the stubborn behaviour, we consider the following assumption:
 \begin{assump}
 \label{Assump:1}
 Each non-stubborn agent in $\mathcal{G}$ has a path from at least one stubborn agent.
 \end{assump}
 
From \cite{friedkin1990opinions}, if a network satisfies Assumption \ref{Assump:1}, 
the final opinions $\mathbf{x}^*$ evolving under the FJ model \eqref{eq:opinion} converge to,
\begingroup
\setlength{\abovedisplayskip}{2pt}
\setlength{\belowdisplayskip}{2pt}
\begin{equation}
\label{eq:opinion_final}
\mathbf{x}^*=(I_n-(I_n-\beta)W)^{-1}\beta \mathbf{x}(0)
\end{equation}
\endgroup
In such networks, only the initial opinions of stubborn agents contribute to the final opinions of the whole group. Consequently, the stubborn agents are \textit{influential}. The impact of each stubborn agent on $\mathbf{x}^*$ is quantified by the influence centrality measure, proposed in \cite{Community_Cleavage}.
Mathematically, the influence centrality vector is defined as:\begingroup
\setlength{\abovedisplayskip}{2pt}
\setlength{\belowdisplayskip}{2pt}
\begin{align}
\label{eqn:influence_centrality}
    \mathbf{c}=\frac{((I_n-(I_n-\beta)W)^{-1}\beta)^T \mathbf{1}_n}{n}
\end{align}
\endgroup
where {$\mathbf{c}=[c_1,c_2,...,c_n]$ and}  $\mathbf{c}^T\mathbf{1}_n=1$. {Only the influence centralities of stubborn agents are non-zero. The magnitudes depend on the network topology ($W$) and the stubbornness of agents ($\beta$).} 




\subsection{Edge modifications}
\label{subsec:Edge_modfications}
An edge modification $(a,b,d)$ with distinct nodes $a,b$ and $d$ comprises the following steps:
\begin{itemize}
    \item the addition of edge $(a,b)$ with edge weight $w\in(0,1)$ (or increase in the edge weight of $(a,b)$ by $w$ such that its weight remains less than $1$ )
    \item the reduction of edge weight of an existing edge $(d,b)$ by $w$.
\end{itemize}
The second step {maintains} the row-stochasticity of $W$ {after} the edge modification. Additionally, choosing $w$ such that $w<w_{bd}$ prevents the edge $(d,b)$ from being removed or having a negative weight. 

\begin{remark}
A recommender system  ranks posts on a user's feed to optimise user engagement and to promote sponsored content \cite{feed_algo}. Sometimes the algorithm suggests posts from new sources or prioritises content from certain sponsored sources. Since users' attention and feed capacity are limited, naturally, certain posts get low visibility. 
An edge modification $(a,b,d)$ captures this phenomenon as the addition of edge $(a,b)$ represents a post suggested to the user, and the reduction in weight of $(d,b)$ models the reduced priority given to other posts on the user's feed.

\end{remark}
\begin{remark}
An edge addition protocol, akin to the proposed $(a,b,d)$ modification, is used to increase the influence centrality of an agent in \cite{wang2025addinglinks}  and to mitigate the impact of malicious agents in \cite{victor}. Under this protocol, an edge $(a,b)$ is added, and the edge weights of all incoming edges at $b$ are reduced. In contrast, {our} edge modification $(a,b,d)$ generalises these approaches by permitting changes in the weights of specific incoming edge at $b$. This flexibility enables us to maintain thematic consistency in recommendations, such that a recommendation that popularises an entity is followed by aligning content, while contradictory information is assigned a lower priority.  
\end{remark} 

Consider a network $\mathcal{G}=(\mathcal{V},\mathcal{E})$ with $m$ stubborn agents such that Assumption \ref{Assump:1} holds.
Let $s$ be a stubborn agent, and {let an} edge modification $(a,b,d)$ be applied to $\mathcal{G}$. The adjacency matrix of the modified graph, denoted by $\hat{W}$, is related to $W$ by a rank-one update:  $\hat {W}=W+we_b(e_a-e_d)^T$. Therefore, the influence centrality of $s$ after the edge modification can be expressed as
\begingroup
\setlength{\abovedisplayskip}{2pt}
\setlength{\belowdisplayskip}{2pt}
\begin{align*}
    \hat{c}_s&=\beta_s \mathbf{1}_{n}^T(I_n-P+\mathbf{u}\mathbf{v}^{T})^{-1} e_{s} /n
\end{align*}
\endgroup
where $P=(I-\beta)W$, $\mathbf{u}=w (1-\beta_b) e_b$, $\mathbf{v}=e_d-e_a$. By the Sherman-Morrison-Woodsbury formula \cite{sherman1950}, we evaluate the inverse $(I_n-P+\mathbf{u}\mathbf{v}^{T})^{-1}$ and determine the change in the influence centrality of $s$, (\textit{i.e.}, $\Delta c_s=\hat{c}_s-c_s$,) is given as
\begingroup
\setlength{\abovedisplayskip}{2pt}
\setlength{\belowdisplayskip}{2pt}
\begin{align}
\label{eqn:delta_c_i}
    \Delta c_s=-\beta_s\frac{ (\mathbf{1}_{n}^TF\mathbf{u})(\mathbf{v}^{T}Fe_s)}{n(1+\mathbf{v}^{T}F\mathbf{u})}
\end{align}
\endgroup
with $F=(I_n-P)^{-1}$.
Clearly, $\Delta c_s=0$ if  $\beta_b=1$. Therefore, each edge modification chosen to enhance influence centrality satisfies the following assumption.
\begin{assump}
\label{assump:2}
  The stubbornness of node $b$ selected for the edge modification $(a,b,d)$ must be strictly less than $1$.
\end{assump}

The opinion of an agent with stubbornness equal to $1$ remains fixed and is not affected by social interactions. However, such agents are rare because the opinions of most individuals are shaped by their peers. Therefore, Assumption \ref{assump:2} has a negligible impact on the number of feasible edge modifications.

%
\subsection{Problems of interest}
\label{Sec:PI}
In this work, our primary objective is to increase the influence centrality $c_s$ {of agent $s$} using edge modification $(a,b,d)$. 
To this end, we address the following {problems.}

\textbf{Problem 1:} \textit{What choice of nodes $a,b$ and $d$ in the network leads to a positive $\Delta c_s$ 
    when edge modification $(a,b,d)$ is applied to $\mathcal{G}$?}

  Unlike \cite{wang2025addinglinks,ancona2022model,gt_attract},  we examine edge modifications $(a,b,d)$ that increase the influence centrality of any desired stubborn agent and are applicable to any weakly connected digraph.
   
\textbf{Problem 2:} \textit{Can certain topological properties of the network of nodes $a,b$ or $d$ ensure that $(a,b,d)$  always results {in} positive $\Delta c_s$?}

 {Note that $\Delta c_s$ depends in general on parameters such as edge weights and the stubbornness of agents.} We seek to characterise nodes $a,b$ and $d$ such that they lead to positive $\Delta c_s$ regardless of the FJ model parameters.

{The following problem is a novel formulation of the classical \textit{influence maximisation} problem in the literature \cite{gionis2013opinion}, where the influence centrality of a stubborn agent is increased using edge modifications. 
} 

\textbf{Problem 3} \textit{How to identify $N$ edge modifications of the form $(a,b,d)$ that yield the maximum increase in the influence centrality of a stubborn agent?}

\section{Edge modifications and  Influence Centrality}
\label{sec:influence_centrality}
In this section, we present the impact of introducing a single edge modification $(a,b,d)$ in $\mathcal{G}$ on the influence centrality measure. More importantly, our focus is on identifying those edge modifications that increase the influence centrality of a stubborn agent.

\begin{thm}
\label{Lemma:path-weights}
Consider a network $\mathcal{G}$ with $m>1$ stubborn agents such that Assumption \ref{Assump:1} holds. Let the opinions evolve under the FJ model \eqref{eq:opinion}. An edge modification $(a,b,d)$, satisfying Assumption \ref{assump:2}, increases the influence centrality $c_s$ of a stubborn agent $s\in \mathcal{V}$ if and only if 
\begingroup
\setlength{\abovedisplayskip}{2pt}
\setlength{\belowdisplayskip}{2pt}
\begin{align}
\label{eqn:increase_inf_cent}
   (e_a-e_d)^T \sum_{k=0}^{\infty}(P)^{k} e_s>0
\end{align}
\endgroup
where $P=(I_n-\beta)W$.
     \end{thm}
     
The proof is presented in Appendix-A. 

{Theorem \ref{Lemma:path-weights} gives a solution to \textbf{Problem 1}. It states} that whether an edge modification with $\beta_b<1$ increases the influence centrality $c_s$ or not depends only on the walks from the stubborn agent $s$ to nodes $a$ and $d$ in the associated graph $G(P)$. 
Clearly, under Assumption \ref{assump:2}, the choice of $b$ is not of much consequence in increasing $c_s$. Consequently, our primary objective is  hereafter to identify suitable nodes $a$ and $d$ such that \eqref{eqn:increase_inf_cent} holds. 

 To verify if a pair of nodes $a$ and $d$ satisfy  \eqref{eqn:increase_inf_cent}, one must identify each walk in $G(P)$
from $s$ to nodes $a$ and $d$ and determine the walk weights.
Thus, finding suitable edge modifications is computationally challenging. {
To simplify this, we seek a solution to \textbf{Problem 2} to characterise pairs that always satisfy \eqref{eqn:increase_inf_cent} for a given stubborn agent $s$. Thereby, simplifying the search for suitable edge modifications.}

\subsection{LTP nodes}
\label{subsec:LTPs}
In \cite{shrinate2025opinionclusteringfriedkinjohnsenmodel}, we introduced a class of agents referred to as LTP agents, who ensure the formation of opinion clusters under the FJ framework. In this work, we examine their topological properties towards increasing the influence centrality of a desired stubborn agent. An LTP agent is defined as follows:  

\begin{definition}[\hspace{-0.1pt}\cite{shrinate2025opinionclusteringfriedkinjohnsenmodel}] \label{defn:LTP}
An agent $p\in \mathcal{V}$ is called an {LTP agent}\footnote{This definition is suitably modified from \cite{shrinate2025opinionclusteringfriedkinjohnsenmodel} according to the framework of the present work, where owing to the Assumption \ref{Assump:1}, only stubborn agents are influential.} if the following conditions hold:
\begin{enumerate}
    \item[(i)] there exists a non-stubborn agent $q$ such that every path in $\mathcal{G}$ from each stubborn agent $s\in \mathcal{V}$ to $q$ (where $p \neq q$) contains $p$ and,
   
    \item[(ii)] if $p$ is non-stubborn, then an agent $c\in \mathcal{V}$ such that all the paths from the stubborn agents to $p$ contains $c$, does not exist.
\end{enumerate}
If $p$ is an LTP agent, then $q$ satisfying condition (i) is said to be {{persuaded by}} $p$. The set of all non-stubborn agents persuaded by $p$ is denoted by $\mathcal{N}_p$.
\end{definition}

If $p$ is a stubborn agent, condition (i) is sufficient to ensure that $p$ is an LTP agent. Condition (ii) is applicable only if $p$ is non-stubborn.  A non-stubborn agent $p$, persuaded by an LTP agent $c$ (\textit{i.e.,} $p\in \mathcal{N}_{c}$), cannot be an LTP agent. 

{Under the FJ model, an LTP agent and its persuaded nodes always converge to the same final opinion \cite{shrinate2025opinionclusteringfriedkinjohnsenmodel}. This indicates that these nodes are closely connected.} 
Building on this notion, we classify the agents in $\mathcal{G}$ based on whether they are persuaded by a stubborn agent $i \in \mathcal{V}$ or not.
\begin{definition}(\textbf{Endorser and Non-endorser})
{Let} a stubborn agent $p$ in $\mathcal{G}$ {be} an LTP agent. {Then,} each agent in $\mathcal{N}_p$ is called an \textbf{endorser} of $p$. {E}ach stubborn agent is also its own endorser. Any agent in $\mathcal{G}$ which is not an endorser of $p$ is called {a} \textbf{non-endorser} of $p$. 
\end{definition}

Note that if a stubborn agent is not an LTP agent, then it has only itself as its endorser. Hence, the set of endorsers of a stubborn agent is always non-empty.  The endorsers of a stubborn agent (other than itself) are non-stubborn. Consequently, in a network with two or more stubborn agents, the set of non-endorsers of each stubborn agent is also non-empty. 


\begin{expm}
\label{expm:LTP}
Consider the network shown in Fig. \ref{fig:14}. Throughout this paper, red nodes represent stubborn agents and orange nodes represent non-stubborn agents. Fig. \ref{fig:34} highlights (using solid lines) that each path from stubborn agents $1$ and $5$ to agents $3$ and $4$ contains agent $2$. Thus, $2$ is an LTP agent with $\mathcal{N}_2=\{3,4\}$. Since agent $2$ is also a stubborn agent, the agents $2-4$ are its endorsers, and the remaining agents are its non-endorsers.

Fig. \ref{fig:15} demonstrates that the stubborn agent $5$ is an LTP agent, with agents $5$ and $6$ are its endorsers.

\end{expm}
 
\begin{figure}[h]
    \centering
    \begin{subfigure}{0.12\textwidth}
        \centering
    \includegraphics[width=1\linewidth]{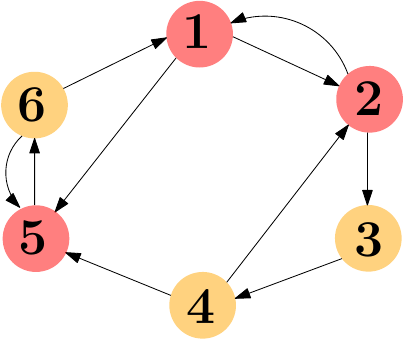}
    \caption{Network $\mathcal{G}$}
    \label{fig:14}
    \end{subfigure}
    \hfill
    \begin{subfigure}{0.14\textwidth}
        \centering
    \includegraphics[width=00.95\linewidth]{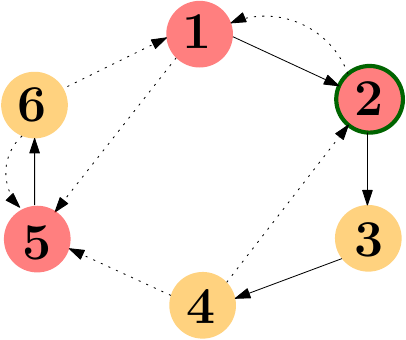}
    \caption{Agent $2$ is LTP}
    \label{fig:34}
    \end{subfigure}
    \hfill
    \begin{subfigure}{0.14\textwidth}
        \centering
    \includegraphics[width=0.95\linewidth]{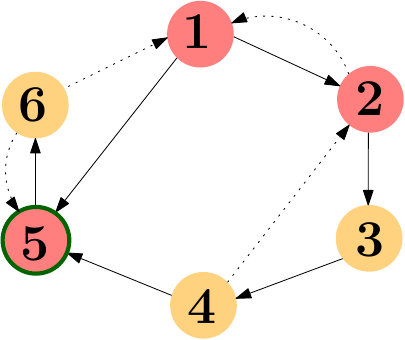}
    \caption{Agent $5$ is LTP}
    \label{fig:15}
    \end{subfigure}
    \caption{ {Illustration of stubborn LTP agents in $\mathcal{G}$ and their endorsers.} }
    \label{fig:redundant_modification_expm}
\end{figure}

\subsection{The relation between network topology and influence centrality}
\label{subsec:KRIC}

In this subsection, we develop a framework to characterise the effect of an edge modification in the network on the influence centrality measure. Using this framework, we establish conditions under which edge modifications increase the influence centrality of a stubborn agent

Let $\mathcal{G}$ be a network with $n$ agents, out of which $m>1$ are stubborn. The nodes in $\mathcal{G}$ can be indexed such that {agents} $1,...,m$ are stubborn and the rest are non-stubborn. At steady state, the following equations hold:
\begingroup
\setlength{\abovedisplayskip}{2pt}
\setlength{\belowdisplayskip}{2pt}\begin{equation}
\label{eqn:final_op}
\begin{split}
   & \mathbf{x}^* = (I_n-\beta)W \mathbf{x}^*+\beta \mathbf{x}(0)  \\
 & \bar{x}=\mathbf{1}_n^T \mathbf{x}^*/n.
\end{split}
\end{equation}
\endgroup
Here, $\bar{x}$ is the average final opinion. From  \eqref{eq:opinion_final} and \eqref{eqn:influence_centrality}, it follows that $\mathbf{c}$ quantifies the impact of initial opinions of the stubborn agents on $\bar{x}$. Further, eqn. \eqref{eqn:final_op} relates $\bar{x}$ with the initial opinions, network topology ($W$) and the stubborn behaviour ($\beta$). 
We can rewrite eqn. \eqref{eqn:final_op} as
\begingroup
\setlength{\abovedisplayskip}{2pt}
\setlength{\belowdisplayskip}{2pt}
\begin{align}
\label{eqn:steady_state}
\underbrace{\begin{bmatrix}
 I_n-(I_n-\beta)W & -\beta[[n],[m]] & \mathbf{0} \\
    \mathbf{0} & \mathbf{0} & \mathbf{0}\\
    -\mathbf{1}_n^T/n & \mathbf{0} & 1  
\end{bmatrix}}_{R} \
\underbrace{
\begin{bmatrix}
    \mathbf{x}^*\\
    \mathbf{x}_s(0) \\
    \bar{x}
\end{bmatrix}}_{\mathbf{z}}=\mathbf{0}
\end{align}
\endgroup
where, $\mathbf{x}_s(0)=[x_{1}(0),x_{2}(0),...,x_{m}(0)]$ denotes the initial opinions of stubborn agents. 
{Interestingly, $R$ in \eqref{eqn:steady_state}, possesses the following properties:}
\begin{enumerate}
    \item $R\mathbf{1}=\mathbf{0}$ 
    \item {the diagonal entries of $R$ are non-negative and its off-diagonal entries are non-positive. }
\end{enumerate}

By Definition 6.3  in \cite{bullo}, $R$ is a Laplacian matrix. Next, we present the construction of the associated network $G(R)$ and its salient properties.
The network $G(R)$ consists of $n+m+1$ nodes and the edges have edge weights equal to the entries in $R$. As $R$ is a Laplacian matrix, all edges in $G(R)$ except self-loops have negative edge weights. Importantly, if $\beta_i<1$, an edge $(j,i)$ in $\mathcal{G}$ contributes an edge $(j,i)$ in $G(R)$ for all $i,j\in [n]$. The nodes in $G(R)$ are associated with the states in $\mathbf{z}$ as follows:
 \begin{enumerate}
     \item each node $j \in \{1,...,n\}$ is associated with the final opinions ${x}_j^*$,
     \item the nodes $\{n+1,...,n+m\}$ are associated with the initial opinion of the $m$ stubborn agents,
     \item Finally, $\{n+m+1\}$ is associated to $\bar{x}$.
 \end{enumerate}
 Note that each stubborn agent $j\in[m]$ in $\mathcal{G}$ contributes two nodes in ${G}(R)$: node $j$ associated with its final opinion and node $n+j$ associated with its initial opinion. We denote the node $n+j$ as $S_j$ and $n+m+1$ as $O$. Each $S_j$ forms a source in $G(R)$ and has an  
 outgoing edge only to the node $j$. The node $O$ has an incoming edge from each node in the set $[n]$. {Let us illustrate with an example.}

\begin{expm}
\label{expm:G_G_a}
Consider the network $\mathcal{G}$ in Fig. \ref{fig:network} with three stubborn agents $1,2$ and $3$ with $\beta_j<1$ for $j\in[3]$. 
The corresponding graph $G(R)$ is shown in Fig. \ref{fig:augmented_graph}, where the additional nodes $S_j$ for and $O$ are highlighted in turquoise.  
As noted earlier, for each edge $(i,k)$ in $\mathcal{G}$, there is a corresponding edge $(i,k)$ in $G(R)$. Further, each $S_j$ has an outgoing edge only to $j$ (for $j\in\{1,2,3\}$), and each node in $\{1,...,6\}$ has an outgoing edge to $O$. The edges to $O$ are represented by dashed lines and the self-loops in $G$ are omitted in Fig. \ref{fig:augmented_graph}.
\end{expm}
\begin{figure}[h]
    \centering
    \begin{subfigure}{0.17\textwidth}
    \centering
\includegraphics[width=0.8\linewidth,height=2.5cm,keepaspectratio]{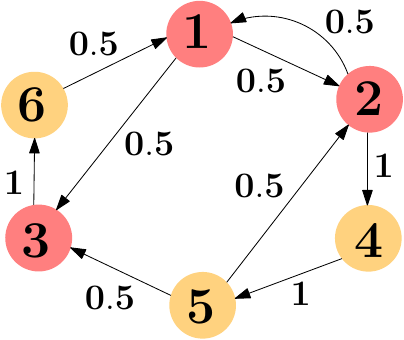}
    \caption{Network $\mathcal{G}$}
    \label{fig:network}
    \end{subfigure}
    \begin{subfigure}{0.2\textwidth}
        \centering
\includegraphics[width=0.8\linewidth,height=2.5cm,keepaspectratio]{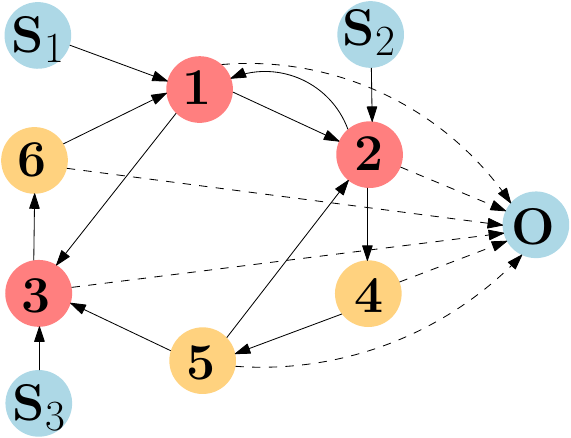}
    \caption{Network $G(R)$}
    \label{fig:augmented_graph}
    \end{subfigure}
    \caption{{Illustration of $\mathcal{G}$ and $G(R)$}}
    \label{fig:graph_and_augmented_graph}
\end{figure}

Since $R$ is a Laplacian matrix, we can reduce ${G}(R)$ by Kron reduction. By reducing $G(R)$, we aim to formalise the impact on the topological properties of nodes $a$ and $d$ on $\bar{x}$  {and influence centrality measure}.
Next, we present some useful properties of the reduced graph $G_{\alpha}$ and the Kron-reduced matrix $R/\alpha^c$.

\begin{lemma}
\label{lm:Basic_properties}
Consider a network $\mathcal{G}$ with $m>1$ stubborn agents that satisfies Assumption \ref{Assump:1}. Let the opinions of agents be governed by the FJ model \eqref{eq:opinion}. Under Kron reduction of $G(R)$, the following properties hold:
\begin{enumerate}
    \item The Kron-reduced matrix $R/\alpha^c$ is well-defined if $\alpha \supseteq \{S_1,...,S_m,O\}$ and $\alpha^c=\{1,...,n,S_1,...,S_m,O\}\setminus \alpha$.
    \item $R/\alpha^c$ is a loopless Laplacian matrix. 
    \item In $G_{\alpha}$, each node in set $\alpha\setminus \{S_1,...,S_m\}$ has a path from {at least one source $S_j$, where $j\in[m]$.}
   
\end{enumerate}
\end{lemma}
{The proof of Lemma \ref{lm:Basic_properties} is in Apppendix-A}

Lemma \ref{lm:Basic_properties} presents {three} fundamental properties that reduced graph $G_{\alpha}$ 
and the corresponding Kron-reduced matrix $R/\alpha^c$ possess. 
The following remark shows that the properties in Lemma \ref{lm:Basic_properties} continue to hold if the graph  $G_{\alpha}$ is further reduced using Kron reduction. 

\begin{remark}
Consider an index set $\alpha$ such that $R/\alpha^c$ is well defined.
Since each node in $\alpha \setminus\{S_1,...,S_m\}$ has a path in $G_{\alpha}$ from an $S_j$ for $j\in[m]$, if we further reduce $G_{\alpha}$ to $G_{\gamma}$ for $\gamma \subset \alpha$ and $\gamma^c=\alpha \setminus \gamma$. Then, by the arguments in Lemma \ref{lm:Basic_properties}, $(R/\alpha^c)/\gamma^c$ is well defined if $\gamma \supseteq\{S_1,....,S_m,O\}$. It also follows that each node in $\gamma\setminus\{S_1,...,S_m\}$ has a path in $G_{\gamma}$ from at least one source $S_1,...,S_m$. Further, since $R/\alpha^c$ is a loopless Laplacian matrix, $(R/\alpha^c)/\gamma^c$ is also a loopless Laplacian matrix. Thus, Lemma \ref{lm:Basic_properties} holds under iterative Kron reduction of $G(R)$.
\end{remark}

\begin{expm}
Consider the network $G(R)$ in Fig. \ref{fig:augmented_graph}. We iteratively reduce $G(R)$ as shown in Fig. \ref{fig:iterative_KR} using the properties of Kron reduction discussed in Lemmas \ref{lm:basic_properties_1} and \ref{lm:Basic_properties}. We reduce $G(R)$ as follows:
\begin{itemize}
    \item First, we consider $\alpha_1=\{1,2,4,S_1,S_2,S_3,O\}$, and hence $\alpha_1^c=\{3,5,6\}$. The network $G(R)$ reduces to $G_{\alpha_1}$ shown in Fig. \ref{fig:network_S1}.
    \item Second, we consider $\alpha_2=\{1,S_1,S_2,S_3,O\}$, and hence $\alpha_2^c=\{2,4\}$. The network $G_{\alpha_1}$ reduces to $G_{\alpha_2}$ shown in Fig. \ref{fig:network_S2}.
    \item Third, we consider $\alpha_3=\{S_1,S_2,S_3,O\}$, and hence $\alpha_1^c=\{1\}$. The network $G_{\alpha_2}$ reduces to $G_{\alpha_3}$ shown in Fig.  \ref{fig:final_graph}. 
\end{itemize}
It follows from Lemma \ref{lm:Basic_properties} that $G(R)$ cannot reduce any further. Note that the self-loops are omitted in Fig. \ref{fig:iterative_KR}.
\end{expm}
\begin{figure}[h]
    \centering
    \begin{subfigure}{0.26\textwidth}
    \centering
\includegraphics[width=0.8\linewidth,height=2.25cm,keepaspectratio]{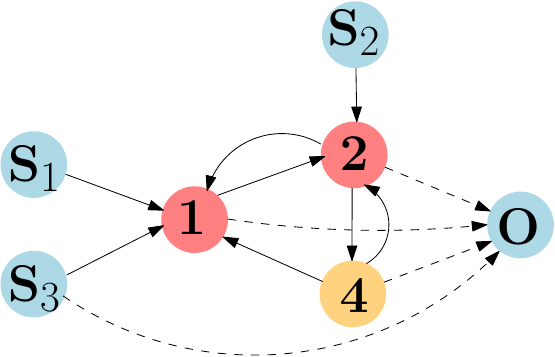}
    \caption{ ${G}_{\alpha_1}$}
    \label{fig:network_S1}
    \end{subfigure}
    \begin{subfigure}{0.2\textwidth}
        \centering
\includegraphics[width=0.7\linewidth,height=2.0cm,keepaspectratio]{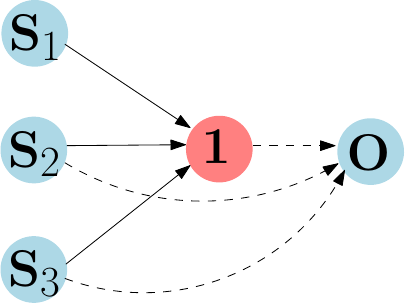}
    \caption{ $G_{\alpha_2}$}
    \label{fig:network_S2}
    \end{subfigure}
       \begin{subfigure}{0.25\textwidth}
        \centering
\includegraphics[width=0.7\linewidth,height=2.0cm,keepaspectratio]{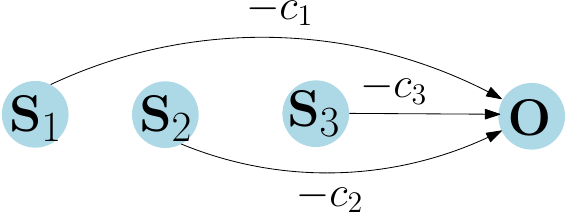}
    \caption{ $G_{\alpha_2}$}
    \label{fig:final_graph}
    \end{subfigure}
 
    \caption{Iterative Kron reduction of  $G(R)$}
    \label{fig:iterative_KR}
\end{figure}

Next, we demonstrate that the Kron reduction of $G(R)$ can be used to determine the influence centrality of the stubborn agents.

\begin{lemma}
\label{lm:influence_cen}
For $\alpha=\{S_1,...,S_m,O\}$, $R/\alpha^c$ is given as
\begin{align}
\label{eqn:Kr_inf_cen}
R/ \alpha^c&=  \begin{bmatrix}
\mathbf{0} & \mathbf{0} \\
 -\mathbf{1}_n^T/n (I_n-(I_n-\beta)W)^{-1} \beta[[n],[m]]& 1
\end{bmatrix} \nonumber \\
&= \begin{bmatrix}
\mathbf{0} & \mathbf{0} \\
 -[c_{1},c_{2},...,c_{m}] & 1
\end{bmatrix}.
\end{align}    
\end{lemma}
\begin{proof}
 For the given $\alpha$, the Kron-reduced matrix 
$R/\alpha^c$ is well defined by Lemma \ref{lm:Basic_properties}.  Since $\alpha^c=[n]$, hence, $R/\alpha^c$ is given by eqn. \eqref{eqn:Kr_inf_cen}.
\end{proof}

Each edge in $G_{\alpha}$ from $S_i$ to $O$ has an edge weight equal to $-c_i$ for all $i\in[m]$. For $G(R)$ in Fig. \ref{fig:augmented_graph}, the reduced graph $G_{\alpha}$ for $\alpha=\{S_1,S_2,S_3,O\}$ is shown in Fig. \ref{fig:final_graph}.   
\subsection{Effect of edge modifications on the influence centrality}
In this subsection, we begin by relating the paths in 
$G(R)$ to the edge weights of the reduced graph $G_{\alpha}$. 
Therafter, we utilise this relation to determine the effect of the edge modifications on the reduced graph.
\begin{lemma}
\label{lm:submatrix_M}
Consider a network $\mathcal{G}$ with $m>1$ stubborn agents such that Assumption \ref{Assump:1} holds. For the corresponding $R$ matrix and $\alpha^c\subseteq [n]$, the submatrix $R[\alpha^c]$ is a non-singular $M$ matrix. 
\end{lemma}
\begin{proof}
Consider $\alpha^c=[n]$, in this scenario $R[\alpha^c]=I_n-(I_n-\beta)W$ is an $M$ matrix due to the following: 
\begin{itemize}
    \item [(i)] it is of the form $sI_n-B$ with positive $s=1$ and non-negative matrix $B=(I_n-\beta)W$. 
    \item [(ii)] $\rho((I_n-\beta)W)<1$. 
\end{itemize}
 Here, (i) holds by definition, 
 (ii) holds because $\mathcal{G}$ satisfies the connectivity condition in Assumption \ref{Assump:1} \cite{friedkin1990opinions}. Thus, $R[\alpha^c]$ is a non-singular $M$ matrix for $\alpha^c =[n]$. Each principal submatrix of a non-singular $M$ matrix is also a non-singular $M$ matrix (Theorem 2.4 of Ch. 6 in \cite{Plemmons}). Thus, $R[\alpha^c]$ is a non-singular $M$ matrix for $\alpha^c\subset[n]$ as well.
\end{proof}

\begin{remark}
\label{rem:walks_edges}
For any $\alpha^c \subseteq [n]$, the matrix $R[\alpha^c]=I-(I-\beta[\alpha^c])W[\alpha^c]$. Lemma \ref{lm:submatrix_M} shows that $R[\alpha^c]$ is a non-singular $M$ matrix. Consequently, $R[\alpha^c]$ has the following important properties:
\begin{itemize}
    \item $R[\alpha^c]^{-1}$ is a non-negative matrix,
    \item $R[\alpha^c]^{-1}= \sum_{k=0}^{\infty}\big((I-\beta[\alpha^c])W[\alpha^c]\big)^k$. 
\end{itemize}
Thus, we can express the Kron-reduced matrix  $R/\alpha^c$ as:
\begingroup
\setlength{\abovedisplayskip}{2pt}
\setlength{\belowdisplayskip}{2pt}
\begin{align}
\label{eqn:walks_paths_o}
    R/\alpha^c=R[\alpha]-R[\alpha,\alpha^c]\sum_{k=0}^{\infty}\big((I-\beta[\alpha^c])W[\alpha^c]\big)^kR[\alpha^c,\alpha]
\end{align}
\endgroup
Eqn. \eqref{eqn:walks_paths_o} relates the weights of the edges in the reduced graph $G_{\alpha}$ with the walks in the original graph $G(R)$. 

Let graph $G_{\alpha}$ be obtained by Kron reduction of $G(R)$ such that $\alpha=\{S_1,...,S_m,O\}$. 
Remark \ref{rem:walks_edges} shows that 
the weight $-c_j$ of edge $(S_j,O)$ in $G_{\alpha}$ 
depends on the weights of all walks from $S_j$ to $O$ in $G(R)$ that traverse the eliminated nodes in set $\alpha^c=[n]$.  

\begin{remark}
{Suppose $G(R)$ is reduced to $G_{\alpha}$ such that $\alpha^c \subset [n]$. By Remark \ref{rem:walks_edges}, the edge weight of $(S_j,O)$ in $G_{\alpha}$ depends on the walks from $S_j$ to $O$ in $G(R)$ that traverse the nodes in $\alpha^c$. Hence, when $\alpha^c \subset [n]$, the walks that determine the edge weight of $(S_j,O)$ in $G_{\alpha}$ are a subset of the walks in $G(R)$ used to determine $-c_j$. Thus, as $\alpha^c \to [n]$, the weight of edge $(S_i,O)$ in $G_{\alpha}$ provides an approximation of the influence centrality $c_i$. }
\end{remark}

{The relation between the influence centrality and walks in $G(R)$ shows that an edge modification $(a,b,d)$ in $\mathcal{G}$ (equivalently in $G(R)$) affects the influence centrality of agents as the walks with edges $(a,b)$ and $(d,b)$ get impacted.
However, whether $(a,b,d)$ increases or decreases the influence centrality of an agent still remains unclear. We address this by relating the topological properties of nodes with the influence centrality. To this end, we iteratively reduce $G(R)$ using Kron reduction 
such that the set $\alpha$ contains only the required nodes. Thereafter, the effect of the modified edges is examined in a significantly reduced graph, which simplifies the analysis. }

The following result establishes how modifying an edge (adding or increasing/decreasing its edge weight) in $G(R)$ affects the reduced graphs obtained under iterative Kron reduction. More generally, we consider any loopy Laplacian matrix $Q$ with the associated graph $G(Q)$ and examine how modifying an edge in $G(Q)$ affects the
reduced graph. 

\end{remark}

\begin{thm}
\label{lemma:modification_propagates}
Consider a loopy Laplacian matrix $Q\in \mathbb{R}^{n \times n}$ with associated digraph ${G}(Q)$. Let $\alpha \subset[n]$ such that $Q/\alpha^c$ is well-defined and $G(Q)$ is reduced to $G_{\alpha}$ using Kron reduction. 
Let $(u,v)$ be an edge in ${G}(Q)$ whose edge weight is increased (or reduced):
\begin{enumerate}
    \item[(i)] If $u,v \in \alpha$, then only the edge weight of $(u,v)$ in $G_{\alpha}$ increases (reduces).
    \item[(ii)] If $u \in \alpha$ and $v \in \alpha^c$, then the edge weight of an edge $(u,j)$ in $G_{\alpha}$ increases (reduces) if there is a path in $G(Q)$ from $u$ to $j$ that passes through $(u,v)$ and traverses only the nodes in $\alpha^c$ (except $u$ and $j$). 
    \item[(iii)] If $u \in \alpha^c$ and $v \in \alpha$, then the edge weight of the edge $(k,v)$ in $G_{\alpha}$ increases (reduces) if there is a path in $G(Q)$ from $k$ to $v$ that passes through $(u,v)$ and traverses only the nodes in $\alpha^c$ (except $k$ and $v$).
    \item[(iv)] If both $u,v \in \alpha^c$, then the edge weight of the edge $(j,k)$ in $G_{\alpha}$ increases (reduces) if there is a path in $G(Q)$ from $j$ to $k$ that passes only through $(u,v)$ and traverses only the nodes in $\alpha^c$ (except $j$ and $k$).
\end{enumerate}
 where $j,k\in \alpha$
\end{thm}

The proof of Theorem \ref{lemma:modification_propagates} is given in Appendix-A.

Theorem \ref{lemma:modification_propagates} specifies the affected edges in the reduced graph $G_{\alpha}$ when the weight of any edge $(u,v)$ is changed in $G(Q)$. Moreover, it also determines whether the edge weight in the reduced graph increases or decreases. Fig. \ref{fig:propagation} illustrates the impact of modifying an edge under each of the conditions presented in Theorem \ref{lemma:modification_propagates}.





\begin{figure}[h]
    \centering
    \begin{subfigure}{0.23\textwidth}
      \centering
    \includegraphics[width=0.8\linewidth]{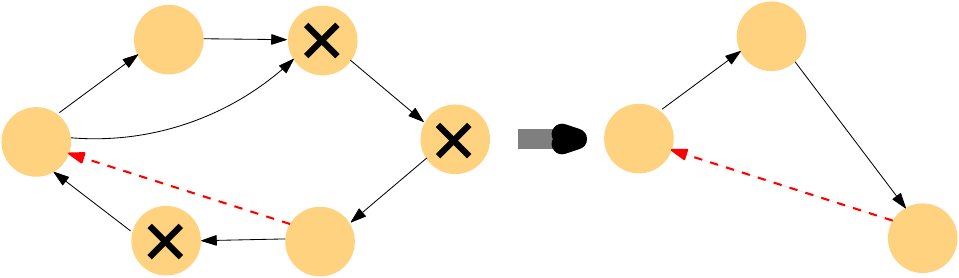}
    \caption*{(i)}
    \label{fig:enter-label}  
    \end{subfigure}
    \begin{subfigure}{0.23\textwidth}
      \centering
    \includegraphics[width=0.8\linewidth]{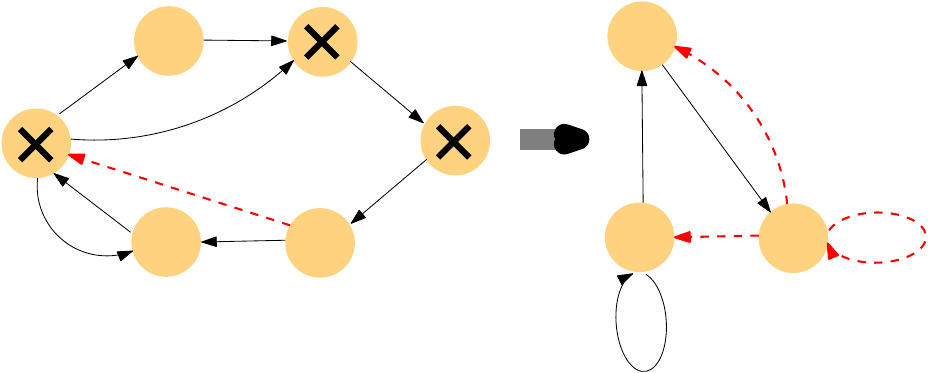}
    \caption*{(ii)}
    \label{fig:enter-label}  
    \end{subfigure}
\begin{subfigure}{0.23\textwidth}
      \centering
    \includegraphics[width=0.8\linewidth]{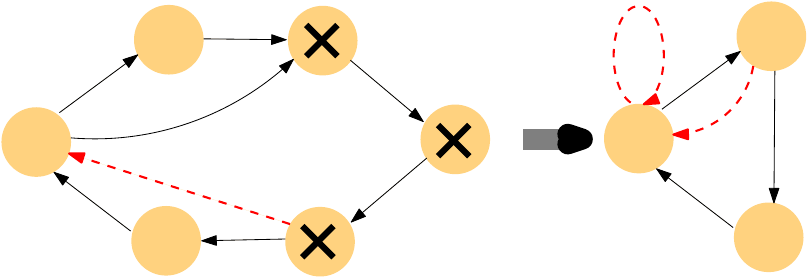}
    \caption*{(iii)}
    \label{fig:enter-label}  
    \end{subfigure}
    \begin{subfigure}{0.23\textwidth}
      \centering
    \includegraphics[width=0.8\linewidth]{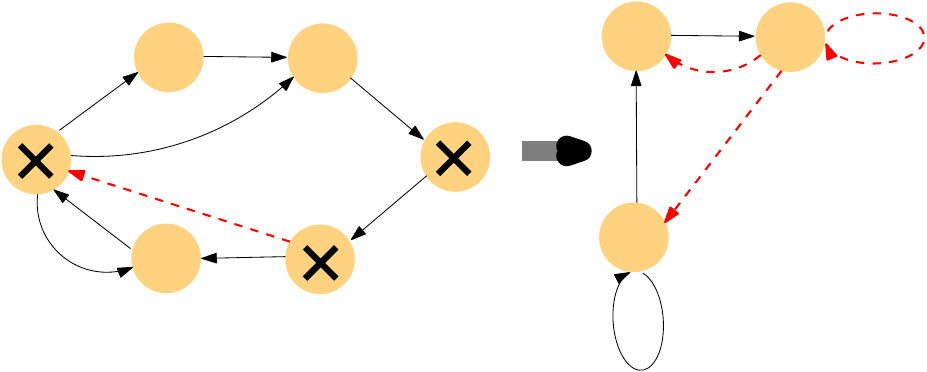}
    \caption*{(iv)}
    \label{fig:enter-label}  
    \end{subfigure}
    \caption{The effect of perturbing an edge in $G(Q)$ on $G_{\alpha}$ for the conditions (i)--(iv) in Theorem \ref{lemma:modification_propagates}. In each subfigure, the network on the left is $G(Q)$ and the reduced network on the right is $G_{\alpha}$. The nodes marked with a cross in $G(Q)$ denote nodes in $\alpha^c$. The perturbed edges in $G(Q)$ and $G_{\alpha}$ are denoted by dashed red arrows.}
    \label{fig:propagation}  
  
\end{figure}

\begin{remark}
In the given setting, $G(R)$ and its reduced graphs are derived from loopless Laplacian matrices. As a result, any modification to $G(R)$ affects at least two or more edges simultaneously to ensure that the modified graph $G(\hat{R}) $ corresponds to a Laplacian matrix $\hat{R}$. Theorem \ref{lemma:modification_propagates} defines the effect of a single edge modification in the $G(R)$ on the reduced graph. When two or more edges are modified in $G(R)$, the change in entries of the Kron-reduced matrix $R/\alpha^c$ must be closely examined.

Suppose $t\geq 2$ edges in $G(R)$ are modified: $(p_1,q_1)$, $(p_2,q_2)$,...,$(p_t,q_t)$.  When $q_k\in \alpha$ for all $k\in [t]$, 
only the submatrices $R[\alpha]$ and $R[\alpha,\alpha^c]$ have modified entries. Hence, we can repeatedly apply Theorem \ref{lemma:modification_propagates} for the $t$ modified edges and determine their effect in the following scenarios:
\begin{itemize}
    \item If $p_k,q_k\in \alpha$ for all $k\in[t]$, by Theorem \ref{lemma:modification_propagates}, the weights of the corresponding edges in $G_{\alpha}$ increase (or reduce).
    \item If there is a node $p_h\in \alpha^c$ for some $h\in [t]$ and  $q_k\in \alpha$ for all $k\in[t]$. The weight of $(j,q_k)$ in $G_{\alpha}$ increases (reduces) only if each path $j\to q_k$ in $G(R)$, composed of  nodes in $\alpha^c \cup \{j,q_k\}$, traverses only the modified edges $(p_i,q_k)$ with increased (reduced) edge weights, where $j\in \alpha$ and $i\in[t]$. Additionally, if the edge $(j,q_k)$ exists in $G(R)$, 
    then its weight must also either increase (reduce) or remain unchanged.
    


\end{itemize}

\end{remark}
Using Theorem \ref{lemma:modification_propagates} and the topological properties of agents presented in Sec. \ref{subsec:LTPs}, in the following subsection, we identify such edge modifications that increase the influence centrality of a desired stubborn agent.
\subsection{Topology-based edge modifications}
\label{sec:useful_modications}
In this subsection, we employ Kron reduction framework to ascertain the impact an edge modification $(a,b,d)$, where associated nodes $a$ and $d$ are from the categories defined in Sec. \ref{subsec:LTPs}. 
We begin by identifying the \textit{redundant edge modifications} under which the influence centrality of each stubborn agent remains unchanged. 

\begin{thm}
\label{thm:redundant_modification} 
Consider a network $\mathcal{G}$ with opinions of agents governed by the FJ model \eqref{eq:opinion}. Let $\mathcal{G}$ contain 
$m>1$ stubborn agents such that Assumption \ref{Assump:1} holds. If $p$ is an LTP agent and $\mathcal{N}_p$ is the set of agents persuaded by $p$, then edge modification $(a,b,d)$, satisfying Assumption \ref{assump:2}, is redundant if $a$ and $d$ belong to the set $\mathcal{N}_p \cup \{p\}  $.
\end{thm}

The proof of Theorem \ref{thm:redundant_modification} is given in Appendix-B.  

Theorem \ref{thm:redundant_modification} shows that if $a,d\in \mathcal{N}_p\cup \{p\}$, then the influence centrality of none of the agents is changed. By Theorem \ref{Lemma:path-weights}, it follows that for the given  $a$ and $d$, the walks from stubborn agent $i\in[m]$ to $a$ and $d$ in $G(P)$ satisfy
\begingroup
\setlength{\abovedisplayskip}{2pt}
\setlength{\belowdisplayskip}{2pt}
\begin{align}
\label{eqn:redundant_md}
e_a^T \sum_{k=0}^{\infty}(P)^{k} e_i=e_d^T \sum_{k=0}^{\infty}(P)^{k} e_i 
\end{align}
\endgroup
for each $i\in [m]$. Interestingly, while $P$ depends on both edge weights and stubbornness, eqn. \eqref{eqn:redundant_md} holds for this choice of $a$ and $d$, independent of the edge weights and the stubbornness of agents. 
Thus, the topological relation of LTP agent $p$ with $a$ and $d$ ensures that 
the corresponding walks in $G(P)$ satisfy eqn. \eqref{eqn:redundant_md}. 

\begin{remark}
 Consider a scenario when $a$ and $d$ are endorsers of a stubborn agent $s$. By Theorem \ref{thm:redundant_modification}, $(a,b,d)$ is a redundant modification because $s$ is an LTP agent with $a,d \in \mathcal{N}_s \cup \{s\}$. Thus, an edge modification, where the edge is added from a stubborn agent, can also be a redundant modification. Hence, the edge-addition protocol from \cite{wang2025addinglinks} does not guarantee a strict increase in influence centrality, as demonstrated in the following example. 
\end{remark}

 \begin{expm}
\label{expm:null_mod}
Consider the network $\mathcal{G}$ in Fig. \ref{fig:14}. In Example \ref{expm:LTP}, the nodes $2$ and $5$ were shown to be LTP agents with $\mathcal{N}_2=\{3,4\}$ and $\mathcal{N}_5=\{6\}$. Consequently, the following edge modifications: (i) $(5,1,6)$ (ii) $(2,5,4)$ and (iii) $(2,4,3)$ are redundant, as illustrated in Fig. \ref{fig:Redundant_Modifications_demo}. 
\begin{figure}[h]
    \centering
    \includegraphics[width=0.7\linewidth,keepaspectratio]{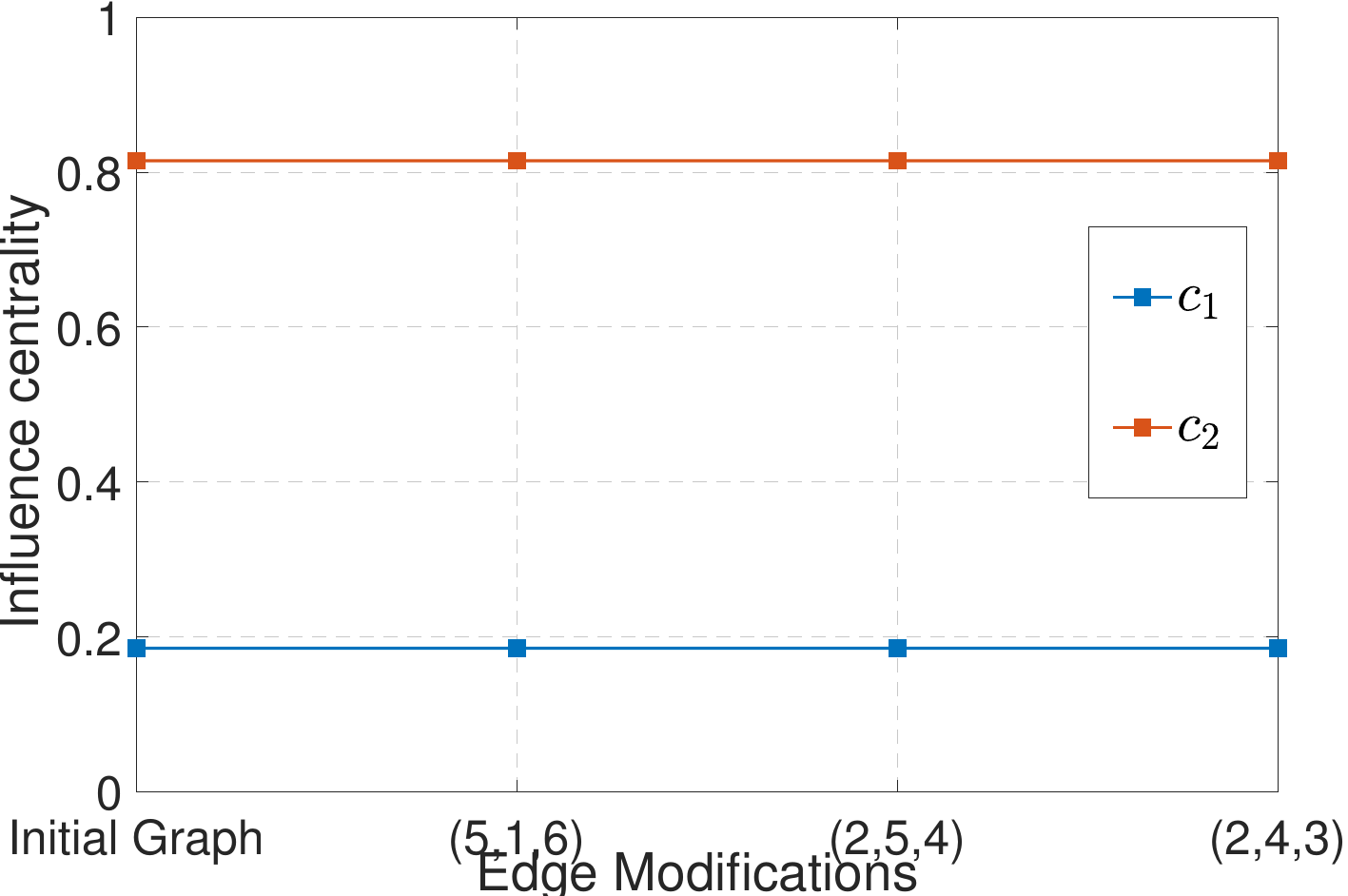}
    \caption{Redundant Modifications in $\mathcal{G}$}
    \label{fig:Redundant_Modifications_demo}
\end{figure}



\end{expm}

Under edge modification $(2,4,3)$, the edge $(2,4)$ with a suitable weight $w$ is added from the stubborn agent $2$. Agent $4$ has only one in-neighbour, whose edge weight is reduced to $(1-w)$. Clearly, this edge modification is equivalent to the edge addition protocol used in \cite{wang2025addinglinks} (where an edge $(a,b)$ with weight $w$ is added from the stubborn agent and the weights of all existing in-neighbours of $b$ are multiplied by $(1-w)$). However, unlike \cite{wang2025addinglinks}, this edge modification is redundant. Thus, in general, adding an edge from a stubborn agent also does not guarantee an increase in its influence centrality. The following result addresses this problem and characterises such edge modifications $(a,b,d)$ that always increase a desired agent's influence centrality in any arbitrary digraph.


\begin{thm}
\label{thm:useful_mod}
Consider a network $\mathcal{G}$ with opinions of agents governed by the FJ model \eqref{eq:opinion}. Let $\mathcal{G}$ contain 
$m>1$ stubborn agents such that Assumption \ref{Assump:1} holds. 
An edge modification $(a,b,d)$, satisfying Assumption \ref{assump:2}, increases the influence centrality of the stubborn agent $s\in[m]$, if $a$ is the endorser of $s$ and $d$ is the non-endorser of $s$.
\end{thm}

The proof of Theorem \ref{thm:useful_mod} is in Appendix-B. 

The endorsers and non-endorsers of a stubborn agent are defined based on graph topology. Thus, a key advantage of Theorem \ref{thm:useful_mod} is that the selection of suitable edge modifications only depends on the network topology and not on the FJ model parameters such as the edge-weights and the stubbornness of agents. {By characterising the topological properties of nodes $a$ and $d$ that guarantee an increase in a stubborn agent's influence centrality under edge modification $(a,b,d)$, Theorem \ref{thm:useful_mod} presents a solution 
to Problem 2.}
\begin{remark}

  Unlike \cite{wang2025addinglinks,ancona2022model}, the proposed condition in Theorem \ref{thm:useful_mod} establishes that a stubborn agent’s influence centrality increases by adding edges from not just itself but also its non-stubborn endorser(s), thereby presenting a more general class of suitable edge modifications. 

The non-stubborn endorsers of a stubborn agent represent supporters (or fan pages in {online social networks}) of influential individuals, such as celebrities or political leaders. 
The edge modifications proposed in Theorem 4 capture the impact of fan pages, which significantly amplify the reach and generate strong word of mouth for the associated influential individual \cite{fan_pages}.

\end{remark}
\begin{remark}
\label{Rem:Maximum_influence_}
From Theorem \ref{thm:redundant_modification}, we know that each pair of endorsers $a$ and $d$ of stubborn agent $s$ satisfy: $e_a^T \sum_{k=0}^{\infty}(P)^{k} e_s=e_d^T \sum_{k=0}^{\infty}(P)^{k} e_s$. Further, eqn. \eqref{eqn:increase_inf_cent} holds for each pair of endorser and non-endorser of a stubborn agent. Hence, $e_a^T \sum_{k=0}^{\infty}(P)^{k} e_s$ achieves the maximum value when $a$ is the endorser of a stubborn agent $s$. 
\end{remark}

As discussed, even verifying whether a single edge modification $(a,b,d)$ increases a desired agent's influence centrality under the condition \eqref{eqn:increase_inf_cent} in Theorem \ref{Lemma:path-weights} is computationally complex. Hence, identifying suitable edge modifications within the network is even more challenging. 
Next, we determine the complexity of identifying suitable edge modifications by employing the endorser property presented in Theorem \ref{thm:useful_mod}. 

Algorithm \ref{algo:identifying_endorsers} outlines the procedure for identifying the endorsers of a stubborn agent $s\in \mathcal{V}$. First, the graph $\tilde{\mathcal{G}}$ is obtained from $\mathcal{G}$ by removing all the outgoing edges of $s$, which takes $O(n)$ time. Next, the graph $\tilde{\mathcal{G}}$ is traversed $m$ times by considering each stubborn agent as the root, using graph traversal techniques such as the Depth First Search (DFS) algorithm. Since a single iteration of DFS has time complexity of $O(n+|\mathcal{E}|)$, the overall traversal of $\tilde{\mathcal{G}}$ takes $O(m(n+|\mathcal{E}|))$ time. Finally, the nodes that are not visited in any of these traversals become the endorsers of $s$ (excluding $s$). Hence, the suitable edge modifications for a desired stubborn agent can be identified in $O(m(n+|\mathcal{E}|))$. Thus, the topology-based condition presented in Theorem \ref{thm:useful_mod} significantly reduces the time-complexity of identifying suitable edge modifications compared to Theorem \ref{Lemma:path-weights}.

\begin{algorithm}[h]
\caption{Identifying the endorsers of  stubborn agent $s$}
\begin{algorithmic}[1]  
    \State Generate graph $\tilde{\mathcal{G}}$ by removing all outgoing edges of $s$ in $\mathcal{G}$.
    \State Define set $\mathcal{H}_j$ for $j\in [m]$
\For{$j= 1$ to $m$}
\State Traverse $\tilde{\mathcal{G}}$ starting from node $j$. 
\State Store each visited node in set $\mathcal{H}_j$
\EndFor
    
    \State Store the nodes visited from each $j\in [m]$ as $\mathcal{H}=\cup_{j=1}^{m} \mathcal{H}_j$.
    \State The nodes $\mathcal{V}\setminus \mathcal{H}$ form the endorsers of $s$ (excluding $s$).
\end{algorithmic}
\label{algo:identifying_endorsers}
\end{algorithm}

\vspace{-15pt}
\section{Influence maximisation using edge modifications}
\label{sec:algorithm}

In the previous section, we identified edge modifications to the network that increase the influence centrality of a stubborn agent. {Building on this result, we examine the influence maximisation problem (\textbf{Problem 3}) in this section.}

Consider a digraph $\mathcal{G}$ with $m>1$ stubborn agents such that Assumption \ref{Assump:1} holds. Let $\mathcal{M}$ denote the set of all possible edge modifications $(a,b,d)$ in the network $\mathcal{G}$. Additionally, for each edge modification $(a,b,d)$, let $w=\zeta w_{b,d}$ with $\zeta \in (0,1)$. This choice of $\zeta$ ensures that the edge $(d,b)$ is not removed from $\mathcal{G}$ under $(a,b,d)$. Thus, the set $\mathcal{M}$ can be denoted as follows:
\begin{align*}
    \mathcal{M}=\{(a,b,d,w)\ | \  a\in\mathcal{V}, \ (d,b)\in \mathcal{E} \text{ and } w=\zeta w_{b,d}\}
\end{align*}
We consider $\zeta$ to be fixed for all edge modifications in $\mathcal{M}$. Next, we formally state the influence maximisation problem using edge modifications $(a,b,d)$.

\textbf{Problem 3} Given a network $\mathcal{G}$ satisfying Assumption \ref{Assump:1} with adjacency matrix $W$ and stubbornness matrix $\beta$. Find a set of $N$ edge modifications $\mathcal{T}\subseteq \mathcal{M}$ that maximises the influence centrality $c_s(\mathcal{T})$ of the stubborn agent $s$. This optimisation problem can be formulated as
\begingroup
\setlength{\abovedisplayskip}{2pt}
\setlength{\belowdisplayskip}{2pt}
\begin{align}
& \operatorname{max}_{\mathcal{T}\subseteq \mathcal{M}} \quad c_s(\mathcal{T}) \nonumber \\
   & s.t. \quad |\mathcal{T}|=N
\end{align}
\endgroup
with $c_s(\mathcal{T})$ denoting the influence centrality of $s$ after the edge modifications in $\mathcal{T}$ are applied on $\mathcal{G}$.

\begin{remark}
Reference \cite{wang2025addinglinks} presents a related influence maximisation problem for a strongly connected network with an external fully stubborn agent ($\beta_s=1$) that does not have a path from the remaining agents. The influence centrality of this external agent is maximised by adding a fixed number of edges from it to the rest. In contrast, the influence maximisation formulated in Problem 3 applies to any weakly (or strongly) connected network. Additionally, the influence centrality of any  stubborn agent can be maximised, regardless of its degree of stubbornness.


\end{remark}



Solving Problem 3 naively is computationally very expensive because we must test all possible subsets $\mathcal{T}$ of $\mathcal{M}$ with cardinality $N$ (\textit{i.e.,} $n|\mathcal{E}| \choose N$). Further, for each subset $\mathcal{T}$, we evaluate $c_s(\mathcal{T})$ that requires computing the inverse $(I-(I-\beta)W)^{-1}$, with  complexity $O(n^3)$. Therefore, it is not feasible to determine the optimal solution of Problem 3. 
In the sequel, we present approximate solutions to Problem 3.

\vspace{-10pt}
\subsection{Greedy heuristic}



Algorithm \ref{algo:optimal_edge_m} outlines a greedy heuristic designed to solve Problem 3. First, it determines $\Delta c_s$ for each edge modification $(a,b,d)$ using eqn. \eqref{eqn:delta_c_i}. Thereafter, it identifies the edge modification in $\mathcal{M}$ that yields the maximum value of $\Delta c_s$ and stores it. Since each edge modification to a digraph alters the entries of $W$, the entries of $F$ are also affected. Thus, $F$ must be updated before selecting the subsequent edge modification. By the Sherman-Woodsbury Formula, the updated matrix $\hat{F}$ due to an edge modification $(a,b,d)$ can be determined as follows: 
\begingroup
\setlength{\abovedisplayskip}{2pt}
\setlength{\belowdisplayskip}{2pt}
\begin{align}
\label{eqn:F_update}
\hat{F}=F  -F\frac{\mathbf{u}\mathbf{v}^T}{1+\mathbf{v}^TF\mathbf{u}} F 
\end{align}
\endgroup
By repeating the selection and update steps $N$ times, the algorithm presents a greedy solution to Problem 3. 
However, this solution is computationally expensive, as shown in the following discussion. 

\begin{algorithm}[h]
\caption{Identification of Edge Modifications by Greedy Heuristic}
\begin{algorithmic}[1]  
\State\textbf{Input:} A network $\mathcal{G}=(\mathcal{V},\mathcal{E})$ with adjacency matrix $W$ and stubbornness $\beta$.
\State \textbf{Output:} A set of $\mathcal{T}\subseteq \mathcal{M}$ edge modifications such that $|\mathcal{T}|=N$.
\State Initialise $\mathcal{T}\leftarrow \emptyset$, $Ve\leftarrow\mathcal{E}$
    \State Evaluate $F$ 
\For{$j = 1$ to $N$}
\State $c_{sm}\gets 0$
        \For{each node $b$ in $\mathcal{V}$}
        \State Calculate $\frac{\mathbf{1}^T F \mathbf{e}_b}{n}$
    \EndFor
    \For{each edge $(d,b)$ in $Ve$}
     \State $w\gets \zeta w_{bd}$
        \For{ each $a\in \mathcal{V}$ such that $a\notin\{b,d\}$}
            \State Calculate value of $\Delta c_s$ using eqn. ~\eqref{eqn:delta_c_i}
            \If{$\Delta c_s>c_{sm}$} 
                 \State $c_{sm}\gets\Delta c_s$, $a_m \gets a,\ b_m\gets b, \ d_m\gets d, \ w_m\gets w$
            \EndIf
        \EndFor
    \EndFor
    \State Update: $\mathcal{T}\gets\mathcal{T} \cup \{(a_m,b_m,d_m,w_m)\}$ and $Ve\gets Ve\setminus\{(d_m,b_m)\}$.
            \State Update: $F$ using eqn. \eqref{eqn:F_update}. 
    \EndFor
\end{algorithmic}
\label{algo:optimal_edge_m}
\end{algorithm}

The computational complexity of Algorithm \ref{algo:optimal_edge_m} is evaluated by analysing the following major steps:
\begin{enumerate}
    \item First, we compute $F$ using power iteration. For a sparse graph with $|\mathcal{E}|$ edges, computing $F$ has a computational complexity of order $O(kn|\mathcal{E}|)$, where $k$ is the total number of power iterations \cite{harm_ful_content}. 
    \item Once $F$ is obtained, we evaluate the term ${1^T F \mathbf{e}_b}$ for each node $b\in \mathcal{V}$. Since computing this term for a single $b$ requires $O(n)$ time, the overall time complexity of this step is $O(n^2)$. 
    \item Next, for each edge $(d,b) \in \mathcal{E}$ and $a \in \mathcal{V}\setminus \{d,b\}$, we compute $\Delta c_s$. Since $F$ and $\mathbf{1}^T F \mathbf{e}_i$ are known, computing $\Delta c_s$ using eqn. \eqref{eqn:delta_c_i} for each triple $(a,b,d)$ requires $O(1)$ time. Therefore, evaluating $\Delta c_s$ for all edge modifications in $\mathcal{M}$ and identifying the one that yields the maximum $\Delta c_s$ requires $O(n|\mathcal{E}|)$ time.
    \item Finally, we update $F$ by eqn. \eqref{eqn:F_update} in $O(n^2)$ time. 
\end{enumerate}

To identify the $N$ edge modifications, Steps 2-4 are repeated $N$ times. Thus, the overall time complexity of Algorithm \ref{algo:optimal_edge_m} is $O((k+N)n|\mathcal{E}|+Nn^2)$.
\vspace{-10pt}

\subsection{Endorser-based heuristic}
\label{subsec:endorser_heuristic}
In this subsection, we develop an efficient and computationally feasible solution to Problem 3 by using the notion of endorsers. 
In Theorem \ref{thm:useful_mod}, we establish that $(a,b,d)$ increases the influence centrality of $s$ if $a$ is its endorser and $d$ is its non-endorser. Moreover, from Remark \ref{Rem:Maximum_influence_}, we know that $\mathbf{e}_a^T F\mathbf{e}_{s}$ attains the highest value when $a$ is an endorser node. Thus, for a given edge $(d,b)$, the term $\mathbf{1}^TF\mathbf{u} \cdot \mathbf{v}^TF\mathbb{e}_s$ (\textit{i.e.,} the numerator of $\Delta c_s$ in eqn. \eqref{eqn:delta_c_i}) attains the maximum value, when $a$ is the endorser of $s$. 

We leverage these properties of endorsers of the desired agent in the proposed heuristic. In this approach, for each selected edge modification $(a,b,d)$, node $a$ is always an endorser of $s$. Consequently, we fix node $a$ as the endorser of $s$ and evaluate $\Delta c_s$ over the edges $(d,b)\in \mathcal{E}$ and select the one for which $\Delta c_s$ attains the maximum value. The steps of the proposed algorithm are presented in Algorithm \ref{algo:greedy_optimal_edge_m} in detail.

\begin{algorithm}[h]
\caption{Identification of Edge Modifications using the endorser-based Heuristic}
\begin{algorithmic}[1]  
\State\textbf{Input:} A network $\mathcal{G}=(\mathcal{V},\mathcal{E})$ with adjacency matrix $W$ and stubbornness $\beta$.
\State \textbf{Output:} A set of $\mathcal{T}\subseteq \mathcal{M}$ edge modifications such that $|\mathcal{T}|=N$.
\State Initialise $\mathcal{T}\leftarrow \emptyset$, $Ve\leftarrow\mathcal{E}$, $a$ be any endsorser of $s$
    \State Evaluate $F$ 
\For{$j = 1$ to $N$}
\State $c_{sm}\gets 0$
        \For{each node $b$ in $\mathcal{V}$}
        \State Calculate $\frac{\mathbf{1}^T F \mathbf{e}_b}{n}$
    \EndFor
    \For{each edge $(d,b)$ in $Ve$}
     \State $w\gets \zeta w_{bd}$
        \State Calculate value of $\Delta c_s$ using eqn. ~\eqref{eqn:delta_c_i}
            \If{$\Delta c_s>c_{sm}$} 
                 \State $c_{sm}\gets\Delta c_s$, $b_m\gets b, \ d_m\gets d, \ w_m\gets w$
            \EndIf
    \EndFor
    \State Update: $\mathcal{T}\gets\mathcal{T} \cup \{(a,b_m,d_m,w_m)\}$ and $Ve\gets Ve\setminus\{(d_m,b_m)\}$.
            \State Update: $F$ using eqn. \eqref{eqn:F_update}. 
    \EndFor
\end{algorithmic}
\label{algo:greedy_optimal_edge_m}
\end{algorithm}

 While Algorithm \ref{algo:optimal_edge_m} evaluates $\Delta c_s$ for all possible triples $(a,b,d)$ in $\mathcal{M}$ to identify the best edge modification in each iteration, Algorithm \ref{algo:greedy_optimal_edge_m} fixes $a$ as the endorser and evalutates $\Delta c_s$ only over the edges $(d,b)\in \mathcal{E}$. Consequently, the computational cost of 
 Step 3 reduces under Algorithm \ref{algo:greedy_optimal_edge_m} to $O(|\mathcal{E}|)$. Since the remaining steps under Algorithm \ref{algo:greedy_optimal_edge_m} remain the same as Algorithm \ref{algo:optimal_edge_m}, the time complexity of Algorithm \ref{algo:greedy_optimal_edge_m} is $O(kn|\mathcal{E}|+N(\mathcal{E}+n^2))$. 
 Hence, Algorithm \ref{algo:greedy_optimal_edge_m} has a significantly reduced run time compared to Algorithm \ref{algo:optimal_edge_m}, especially in large networks. In the sequel, we validate the effectiveness of Algorithm \ref{algo:greedy_optimal_edge_m} by comparing it with suitable baselines.
\section{Empirical Evaluation}
\label{sec:simulations}

\begin{figure}
    \centering
    \begin{subfigure}{0.45\textwidth}
        \centering
        \includegraphics[width=0.8\textwidth]{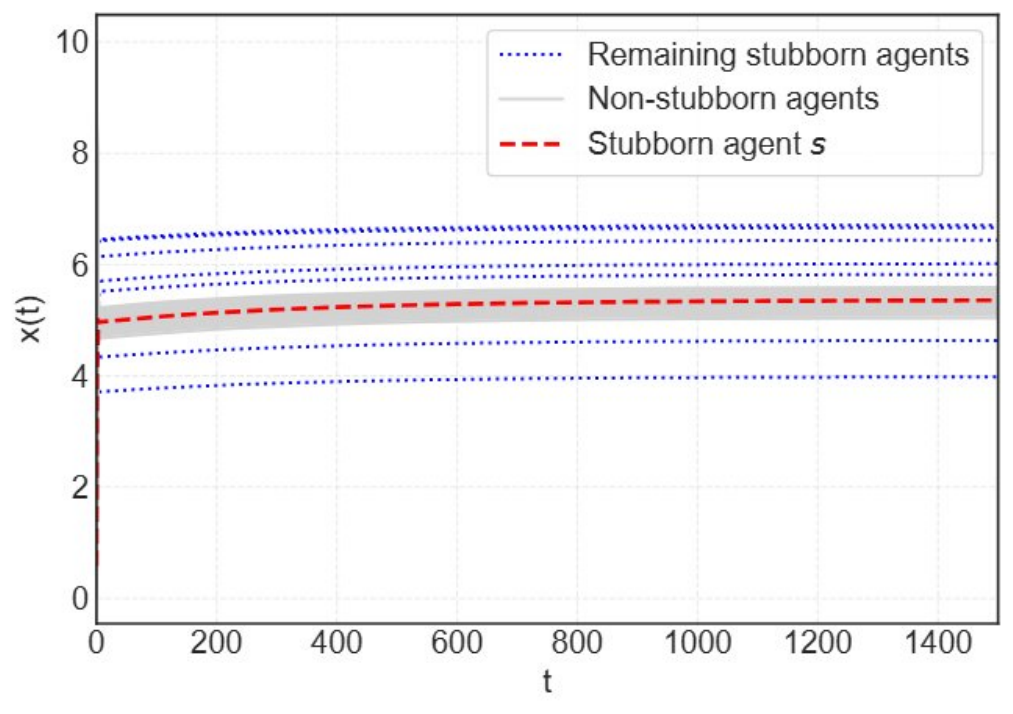}
        \vspace{-5pt}
        \caption{Original Graph}
        \vspace{10pt}
        \label{Fig:opinion_evol_OG}
    \end{subfigure}
    
        \begin{subfigure}{0.45\textwidth}
        \centering
        \includegraphics[width=0.8\textwidth]{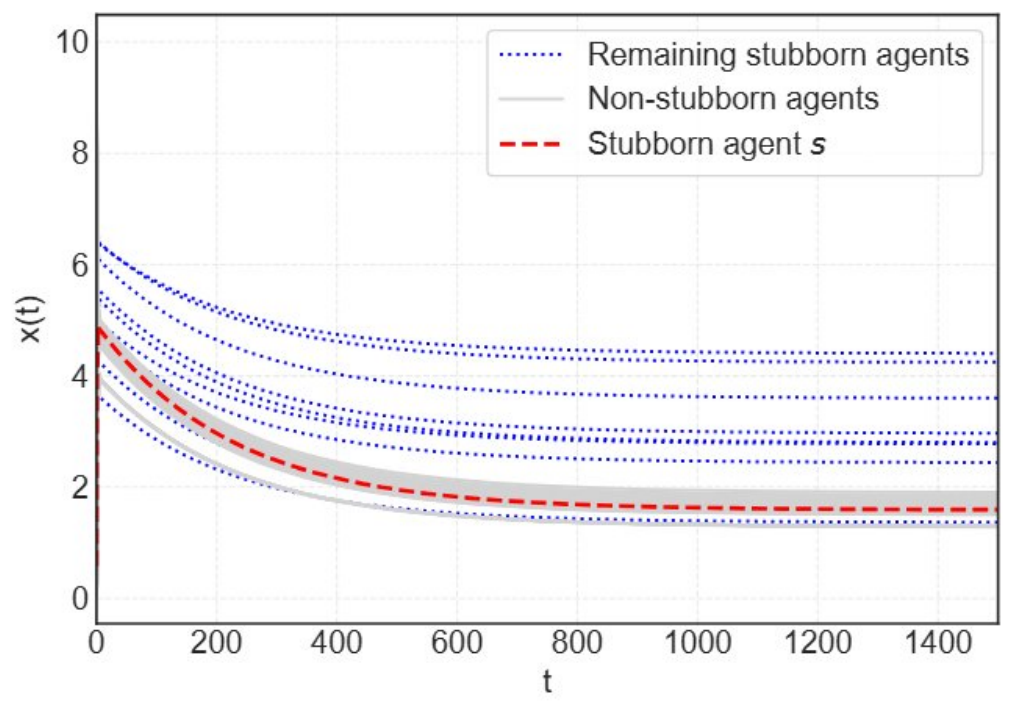}
        \vspace{-5pt}
        \caption{$N=50$}
        \vspace{10pt}
        \label{Fig:opinion_evol_n_50}
    \end{subfigure}
        \begin{subfigure}{0.45\textwidth}
        \centering
        \includegraphics[width=0.8\textwidth]{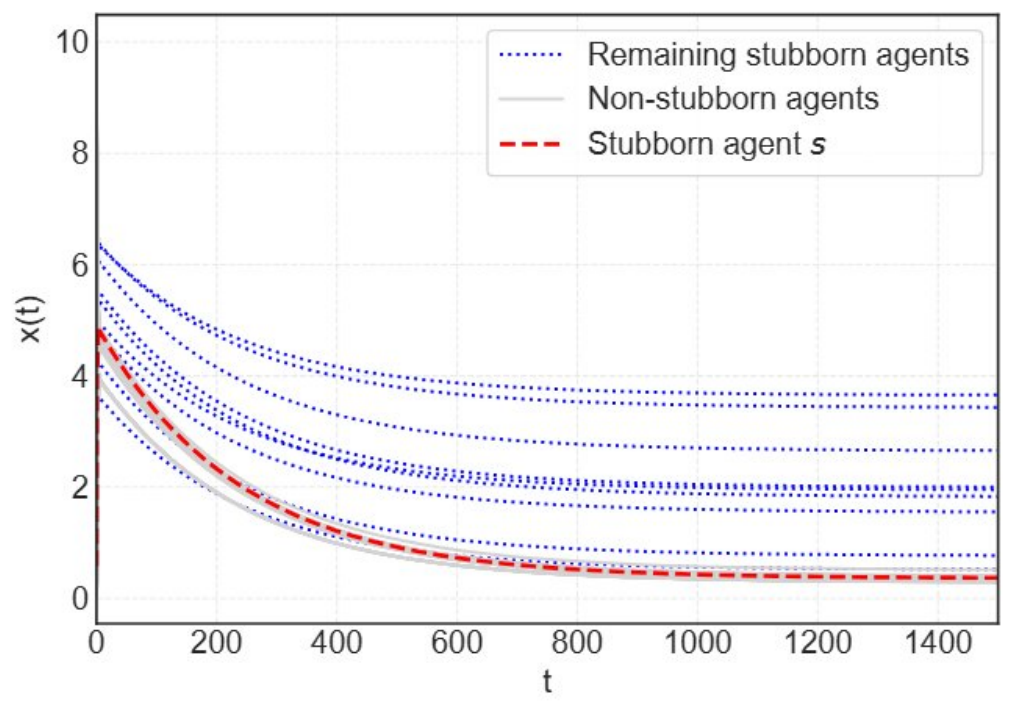}
        \vspace{-5pt}
        \caption{$N=100$}
        \vspace{10pt}
        \label{Fig:opinion_evol_n_100}
    \end{subfigure}
        \begin{subfigure}{0.45\textwidth}
        \centering
        \includegraphics[width=0.8\textwidth]{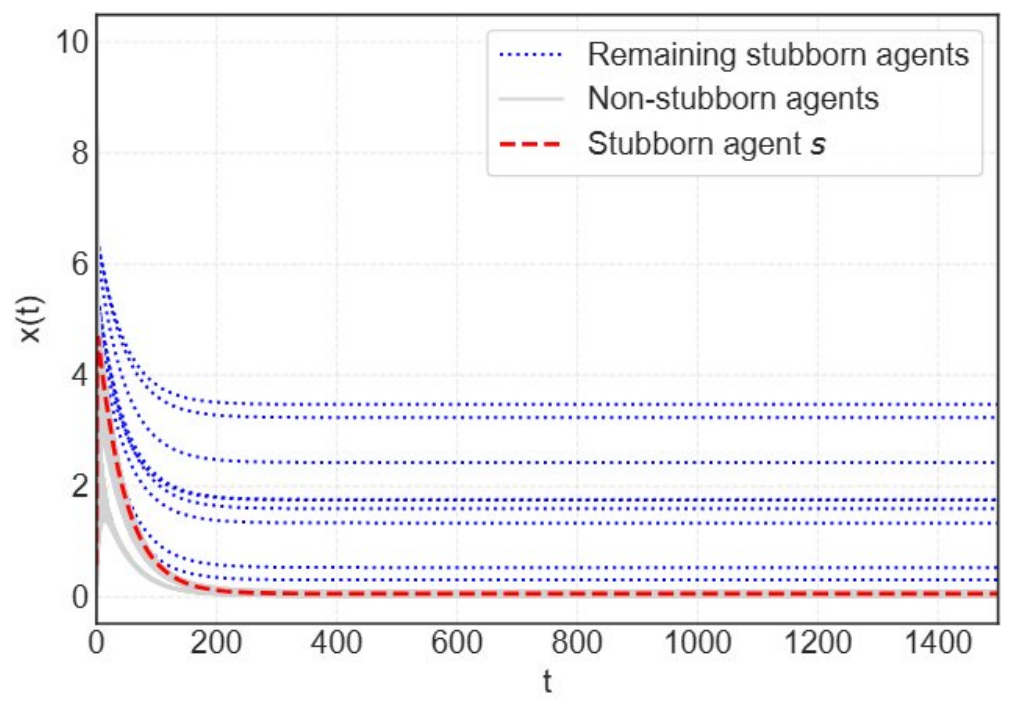}
        \vspace{-5pt}
        \caption{$N=200$}
        \vspace{10pt}
        \label{Fig:opinion_evol_n_200}
    \end{subfigure}
    \caption{Evolution of opinions under the FJ model in the original Erd\H{o}s-R\'{e}nyi graph and the obtained graphs after $N\in\{50,100,200\}$ edge modifications}
\end{figure}

{In this section, we present the impact of increasing a stubborn agent's influence centrality by 
Algorithm \ref{algo:greedy_optimal_edge_m} on opinion evolution under the FJ model. We also test the effectiveness of this approach in increasing an agent's influence centrality and compare it with other baselines.}

{To test the scalability of Algorithm \ref{algo:greedy_optimal_edge_m}, we perform this analysis on an Erdős-Rényi digraph of $1000$ nodes with edge probability \( p=0.0090 \). {Since \(p> {\log n}/{n}=0.0069 \), the digraph is strongly connected with a high probability \cite{graham2008note}}. Out of the $1000$ nodes in the digraph, we select $10$ stubborn agents uniformly at random 
and assign the stubbornness of each stubborn agent from the uniform distribution over $[0.1,0.9]$. The stubborn agent with the least influence centrality in the original graph ($\approx0.1$) is selected as the stubborn agent $s$. We increase the influence centrality of $s$ by adding $N$ edge modifications using Algorithm \ref{algo:greedy_optimal_edge_m}. Throughout all simulations, we consider $\zeta=0.9$.}

\subsection{Impact of edge modifications on opinion formation}
{In this subsection, we show how increasing the influence centrality of stubborn agent $s$ using Algorithm \ref{algo:greedy_optimal_edge_m} affects  opinion evolution under the FJ model \eqref{eq:opinion}. 
To this end, the initial opinions of the agents are selected from the uniform distribution over $(0,10)$ and the initial opinion of $s$ is set to $0$.
First, we plot the evolution of opinions under the FJ model for the original graph in Fig. \ref{Fig:opinion_evol_OG}. Thereafter, we use Algorithm \ref{algo:greedy_optimal_edge_m} to select $N\in\{50,100,200\}$ edge modifications that increase influence centrality of $s$. 
Figs. \ref{Fig:opinion_evol_n_50}-\ref{Fig:opinion_evol_n_200} show the  evolution of opinions under the FJ model in the updated graphs {obtained after
 $N=50, N=100$ and $N=200$ selected edge modifications, respectively.}

{In the original graph, the final opinions of each agent is strictly greater than $0$, as shown in Fig. \ref{Fig:opinion_evol_OG}. However, as $N$ increases, the final opinions of the non-stubborn agents and $s$ converge closer to $0$, see Figs. \ref{Fig:opinion_evol_n_100} and \ref{Fig:opinion_evol_n_200}. The opinions of remaining stubborn agents also shift toward $0$, though less sharply than the non-stubborn agents.}

{This observation follows from the definition of an agent's influence centrality.
By definition, a stubborn agent's influence centrality equals its contribution in the average final opinion $\bar{x}$. As mentioned, the initial opinion of $s$ is $0$. 
Hence, since $s$ has the least influence centrality in the original digraph, the opinions of the agents converge to a value much greater than $0$. However, as $N$ increases under Algorithm \ref{algo:greedy_optimal_edge_m}, $c_s$ also increases and the opinions of all the agents shift closer to $0$. Adhering to their nature, the stubborn agents are more reluctant to shift their opinions closer to $0$ compared to the non-stubborn agents, which almost converge to $0$ for higher values of $N$.}

\subsection{Comparitive analysis of the endorser-based heuristic}
In this subsection, we present the effectiveness of increasing the influence centrality using the endorser-based heuristic Algorithm \ref{algo:greedy_optimal_edge_m}. We compare its performance to the greedy heuristic in Algorithm \ref{algo:optimal_edge_m} and the following baselines:
\begin{itemize}
    \item \textit{Random selection}: We determine the $N$ edge modifications by choosing $N$ edges from $\mathcal{E}$  uniformly at random; and consider $a$ to an endorser of $s$ in each edge modification. It has a time complexity of $O(N)$.
    \item \textit{Top $N$ edge modifications:} We evaluate $\Delta c_s$ for each suitable $(a,b,d)$ in $\mathcal{M}$ for the given $\mathcal{G}$,$W$ and $\beta$. Thereafter, we select $N$ edge modifications that yield the highest values of $\Delta c_s$. It has a time complexity of $O(kn|\mathcal{E}|+n^2)$. 
\end{itemize}


In Random selection baseline, we consider $a$ to always be the endorser of $s$ to ensure a monotonic increase in the influence centrality under the selected $N$ edge modifications. Further, we consider the Top $N$ edge modifications as a baseline to highlight the importance of accounting for the change in the network arising due to previous edge modification while determining the future ones.


For the Erd\H{o}s-R\'{e}nyi digraph $\mathcal{G}$, and stubbornness values $\beta$, we obtain $N$ edge modifications using the greedy algorithm, endorser-based heuristic and the two baselines. We compare the effectiveness of these algorithms in increasing the influence centrality in Fig. \ref{fig:er_influence}. We observe that the performance of Algorithms \ref{algo:optimal_edge_m}, \ref{algo:greedy_optimal_edge_m} and the Top N edge modifications baseline is significantly better than the Random selection. The influence centrality of $s$ approaches $1$ after $140$ edge modifications under both greedy and endorser-based heuristics. 
Although the greedy algorithm performs only slightly better than the endorser-based heuristic, its significantly high run-time makes it unscalable for large networks. The comparison with Top-N edge modifications shows that updating $F$ after each modification results in higher time complexity but significantly improved performance.


\begin{figure}[h]
    \centering
    \includegraphics[width=0.95\linewidth]{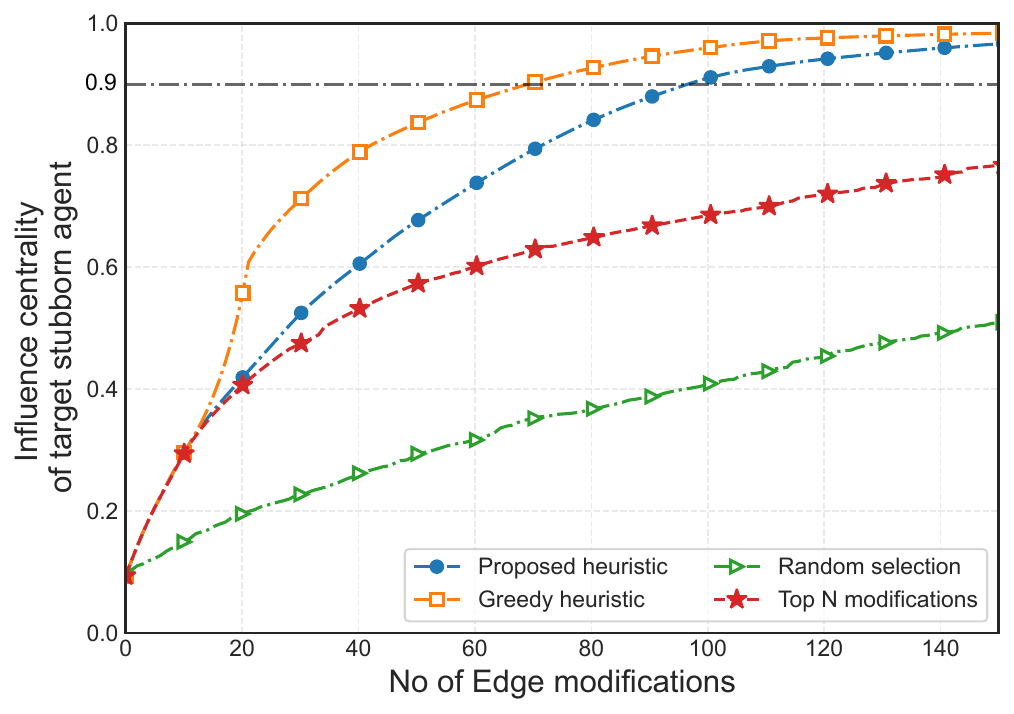}
    \caption{Edge modifications to increase the influence of a randomly chosen stubborn agent}
    \label{fig:er_influence}
\end{figure}  
\section{Conclusion}
\label{sec:conclusion}
In this paper, we address the problem of maximising the influence centrality of a desired stubborn agent in the FJ framework. Towards this objective, we use the edge modifications $(a,b,d)$ that mimic content moderation performed by recommendation algorithms. 
Leveraging the notion of the LTP agents presented in \cite{shrinate2025opinionclusteringfriedkinjohnsenmodel}, we categorise the agents in a network with respect to a particular stubborn agent as: its endorsers and its non-endorsers, based on their topological properties. Thereafter, we construct a network $G(R)$ from the linear equations that define the relation between the final opinions and the initial opinion of the stubborn agents. Using the Kron reduction of $G(R)$, we relate the impact of an edge modification $(a,b,d)$ on the influence centrality vector. Through this framework, we show that $(a,b,d)$ always increases the influence centrality of a stubborn agent if $a$ is its endorser and $d$ is its non-endorser. Interestingly, the influence centrality of the desired stubborn agent increases under this condition, independent of the stubbornness of agents and the edge weights the network.
Moreover, if both $a$ and $d$ are endorsers of the same stubborn agent, then $(a,b,d)$ becomes redundant.  

Based on the proposed edge modifications, we present an endorser-based approximate solution for the discrete optimisation problem that selects $N$ edge modifications to maximise a desired agent's influence centrality. The time complexity of the proposed heuristic is $O(kn|\mathcal{E}|+N(|\mathcal{E}|+n^2))$, which is significantly better than its greedy counterpart with $O((k+N)n|\mathcal{E}|+Nn^2)$. Further, simulations on Erdős–Rényi random graphs show that the proposed heuristic maintains performance as effective as the greedy approach. Hence, it is a computationally efficient alternative to the greedy approach, suitable for large networks.

{In this work, we show that interactions can be suitably modified to enhance influence centrality in the classical FJ model. In future,  
it would be interesting to extend such mechanisms in the non-linear and multi-issue versions of the FJ model.} 

\section*{Appendix-A}

\textbf{\textit{Proof of Theorem} \ref{Lemma:path-weights}:} 
To establish conditions under which $(a,b,d)$ increases $c_s$, we examine $\Delta c_s$ in eqn. \eqref{eqn:delta_c_i}. First, we establish that $\mathbf{1}_{n}^TF\mathbf{u}$ and $(1+\mathbf{v}^{T}F\mathbf{u})$ are always positive under the given conditions.

For any network $\mathcal{G}$ that satisfies Assumption \ref{Assump:1}, $\rho(P)<1$ \cite{friedkin1990opinions}. Hence, we can expand $F$ as: $F=\sum_{k=0}^{\infty}(P)^{k}$ by Neumann series (Ch 2. in \cite{bullo}). Thus, $F$ is a non-negative matrix and $\mathbf{1}_{n}^TF\mathbf{u}$ is positive, when $\beta_b<1$. 
Next, we prove that $1+\mathbf{v}^{T}F\mathbf{u}>0$.

Matrix $I-P$ is row-diagonally dominant and is invertible under Assumption \ref{Assump:1}. Hence, by Lemma \ref{Lemma:RDD_Invertible}, $F$ is diagonally dominant of its column entries \textit{i.e.,}  
$F[j,j]\geq F[i,j]$ for any  $j\in[n]$ and $i\in[n]\setminus\{j\}$.  Additionally, $F$ also satisfies:
\begingroup
\setlength{\abovedisplayskip}{2pt}
\setlength{\belowdisplayskip}{2pt}
\begin{align}
\label{eqn:I-P}
(I_n-P)\begin{bmatrix}
 F[1,j] & F[2,j] & \cdots & F[n,j]
\end{bmatrix}^T&=e_j \quad \forall j\in[n]   
\end{align}
\endgroup

For a given pair $b$ and $d$, we define the following functions: $h(e)=yF[e,b]$  and $g(e)=h(e)-h(d)-1$, where $y=(1-\beta_b)w$ and $e\in \mathcal{V}$. Since $F[j,b]\leq F[b,b]$ for any $j\in \mathcal{V}\setminus \{b\}$, clearly $g(j)\leq g(b)$. To show that $g(j)<0$, we prove that $g(b)<0$.

From eqn. \eqref{eqn:I-P}, we know  $h(b)=y+\sum_{k=1}^np_{bk}h(k)=y+p_{bd}h(d)+\sum_{k=1,k\neq d}^np_{bk}h(k)$, where $p_{ij}$ is the entry of $P$. Further, we can write this eqn. as: 
\begingroup
\setlength{\abovedisplayskip}{2pt}
\setlength{\belowdisplayskip}{2pt}
\begin{align*}
 &g(b)+1+h(d)=y+p_{bd}h(d)+\sum_{k=1,k\neq d}^np_{bk}(g(k)+1+h(d)) \\
 &=y+p_{bd}h(d)+(1-\beta_b-p_{bd})(1+h(d)) +\sum_{k=1,k\neq d}^np_{bk}g(k)
\end{align*}
\endgroup

Hence, we can further simplify as follows:
\begingroup
\setlength{\abovedisplayskip}{2pt}
\setlength{\belowdisplayskip}{2pt}
\begin{align*}
g(b)&=y+p_{bd}h(d)-(\beta_b+p_{bd})(1+h(d))+\sum_{k=1,k\neq d}^np_{bk}g(k)\\
&=y-p_{bd}-\beta_b(1+h(d))+\sum_{k=1,k\neq d}^np_{bk}g(k) \\
&\leq y-p_{bd}-\beta_b(1+h(d))+(1-p_{bd}-\beta_b)g(b)
\end{align*}
\endgroup
The last inequality follows because $g(b)\geq g(j)$ for any $j\in \mathcal{V}\setminus \{b\}$. By definition, $p_{bd}=(1-\beta_b)w_{bd}$. Thus, $g(b) \leq (y-(1-\beta_b)w_{bd}-\beta_b(1+h(d)))/(\beta_b+(1-\beta_b)w_{bd})$. Since ${w}<w_{bd}\leq 1$, $g(b)<0$. Consequently, $g(j)<0$ for all $j\in \mathcal{V}$. Hence, $1+\mathbf{v}^{T}F\mathbf{u}>0$. Thus, the influence centrality $c_i$ increases when eqn. \eqref{eqn:increase_inf_cent} holds.
\hfill$\blacksquare$

\textbf{\textit{Proof of Lemma} \ref{lm:Basic_properties}:}
Due to Assumption \ref{Assump:1}, the first $n$ nodes in $G(R)$ satisfy the following: each non-stubborn node $i$ (node associated with the final opinion of a non-stubborn agent) has a path from a stubborn node $j$ (node associated with the final opinion of a stubborn agent). Moreover, the source $S_j$ has an outgoing edge to node $j$ in $G(R)$ for $j\in [m]$. Hence, each node $i \in [n]$ has a path from a source $S_j$ for $j\in[m]$.  
Thus, for any $\alpha \supseteq \{S_1,...,S_m,O\}$, each node in $\alpha^c\subseteq[n]$ has a path from a node in $\alpha$. By Lemma \ref{lm:basic_properties_1}, the Kron-reduced matrix $R/\alpha^c$ is well defined. 

Statement 2) follows directly from Lemma \ref{lm:basic_properties_1}. For statement 3), we use the fact that each node indexed $\{1,2,...,n\}$ has a path from an $S_j$ in $G(R)$  for $j\in [m]$. Consider a suitable $\alpha$ such that $R/\alpha^c$ is well-defined. 
In general, a path from $S_j$ to $p$ in $G(R)$ of the following form exists: $S_j \to q_1 \to q_2 \to p_1 \to q_3 \to q_4 \to p_2 \to \cdots \to p$ with certain nodes $q_i\in \alpha^c$ and $p_h \in \alpha$. It follows from Lemma \ref{lm:basic_properties_1} that $G_{\alpha}$ has the path $S_j \to p_1 \to p_2 \to \cdots \to p$. Consequently, each $p\in[n]\cap \alpha$ has a path from a source $S_j$ in $G_{\alpha}$ for $j\in[m]$. \hfill$\blacksquare$

\textbf{\textit{Proof of Theorem} \ref{lemma:modification_propagates}:} 
 By Lemma \ref{lm:Q_is_M_matrix}, $Q$ is an M matrix. Clearly, any principal sub-matrix of $Q$ is also an $M$ matrix.
  Hence, for any $\alpha^c$, $Q[\alpha^c]$ is an M matrix. Additionally, if $Q/\alpha^c$ is well-defined, then $Q[\alpha^c]$ is invertible and hence it is a non-singular M-matrix. Let $Q[\alpha^c]=tI-B[\alpha^c]$, here $t>\rho(B[\alpha^c])$. 
Therefore, by Neumann-series, we can re-write the Kron-reduced matrix $Q/\alpha^c$ as:
\begingroup
\setlength{\abovedisplayskip}{2pt}
\setlength{\belowdisplayskip}{2pt}
\begin{align}
\label{eqn:edge_mod_paths}
 Q/\alpha^c=Q[\alpha]+\frac{1}{t}Q[\alpha,\alpha^c](\sum_{k=0}^{\infty}\big(\frac{B[\alpha^c]}{t}\big)^{k})Q[\alpha^c,\alpha]   
\end{align}
\endgroup

Consider the edge weight of $(u,v)$ in $\mathcal{G}$ is increased. By definition of $Q$, the $(v,u)^{th}$ entry of $Q$ gets reduced. Equivalently, in $G(Q)$, the edge weight of $(u,v)$ in ${G}(Q)$ is reduced. We examine the Kron-reduced matrix $Q/\alpha^c$ to determine the effective changes in $G_{\alpha}$. Depending on the sets $\alpha$ or $\alpha^c$ in which  $u$ or/and $v$ belong, the set of edges in $G_{\alpha}$ that are affected differ.

Under $(i)$, both $ u, v \in \alpha$. Therefore, only the submatrix $Q[\alpha]$ has a modified entry due to edge $(u,v)$. By the definition of $Q/\alpha^c$, it follows that only the  $(v,u)^{th}$ entry of $Q/\alpha^c$ gets reduced. This change is equivalently reflected in edge $(u,v)$ in $G_{\alpha}$.

Under (ii), since $u\in \alpha$ and $v\in \alpha^c$, the entry in submatrix $Q[\alpha^c,\alpha]$ gets modified. For
any $\alpha$ such that $Q/\alpha^c$ is well-defined, $Q/\alpha^c$ can be simplified as in eqn. \eqref{eqn:edge_mod_paths}. By definition, the matrix $B[\alpha^c]$ has a positive entry corresponding to each edge $(e,f)$ in $G(Q)$ if $e,f \in \alpha^c$. Consequently, by eqn. \eqref{eqn:edge_mod_paths}, if there is a node $j \in \alpha$ such that a path from $u$ to $j$ exists in ${G}(Q)$ that traverses edge $(u,v)$ and this path only traverses the nodes in $\alpha^c$ (except $u$ and $j$), then the weight of edge $(u,j)$ in $G_{\alpha}$ reduces.

Similarly, under (iii), $u\in \alpha^c$ and $v\in \alpha$, the entry in submatrix $Q[\alpha,\alpha^c]$ gets decreased due to $(u,v)$. Again from eqn. \eqref{eqn:walks_paths_o}, it follows that if there is a node $k\in \alpha^c$ such that there is a path from $k$ to $v$ that traverses $(u,v)$ and this path only passes through nodes in $\alpha^c$ (except $k$ and $v$). Then, the weight of the edges $(k,v)$ in $G_{\alpha}$ for $k\in \alpha$ decreases.

Finally, when both $u,v\in \alpha^c$, then the entry of the submatrix $Q[\alpha^c]$  gets modified. From eqn. \eqref{eqn:walks_paths_o}, it follows that for any $j,k\in \alpha$, if there is a path in $G(Q)$ that traverses edge $(u,v)$ and only nodes in $\alpha^c$, the edge weight of $(j,k)$ in $G_{\alpha}$ reduces. \hfill$\blacksquare$


\vspace{-10pt}
\section*{Appendix-B}
\textbf{\textit{Proof of Theorem} \ref{thm:redundant_modification}:} 
Suppose $a,d \in \mathcal{N}_p$. Under an edge modification $(a,b,d)$ in $\mathcal{G}$, the edge weight of $(a,b)$ decreases by $-(1-\beta_b)w$, while that of $(d,b)$ increases by $+(1-\beta_b)w$ in $G({R})$ to form $G(\hat{R})$ (The matrix $\hat{R}$ is derived from eqn. \eqref{eqn:steady_state} for modified $\mathcal{G}$). 
If there exists a node set $\alpha \supseteq \{S_1,...,S_m,O\}$ such that the reduced graphs  $G_{\alpha}$ and $\hat{G}_{\alpha}$ obtained from $G(R)$ and ${G}(\hat{R})$, respectively, are identical. Then, further reducing $\hat{G}_{\alpha}$ to $\hat{G}_{\eta}$ such that $\eta=\{S_1,...,S_m,O\}$ yields the same influence centrality measure as the original network. Hence, an edge modification $(a,b,d)$ is redundant if such a node set $\alpha$ exists.
To determine the suitable $\alpha$, we perform the following reductions:

\textbf{Step 1:} Let $\alpha_1=\{b,p,S_1,...,S_m,O\}$. By Lemma \ref{lm:Basic_properties}, the Schur complement $R/\alpha_1^c$ is well-defined. %
Under the edge modification $(a,b,d)$, the edges $(a,b)$ and $(d,b)$ in $G(R)$ are modified and the nodes $a,d \in \alpha_1^c$ and $b \in \alpha_1$. Thus, by Theorem \ref{lemma:modification_propagates}, only the edge weight of $(j,b)$ in $G_{\alpha_1}$ is modified (increased or decreased) for $j\in \alpha_1$. The change in the edge weight occurs only if a path from $j$ to $b$ exists in $G(R)$ such that each node on this path (other than $j$ and $b$) is from set $\alpha_1^c$ and it traverses either $(a,b)$ or $(d,b)$ or both. 

Recall that $p$ is an LTP agent and the nodes $a,d \in \mathcal{N}_p$. Hence, a path from any stubborn agent $i\in[m]$ to $a$ and to $d$ in $\mathcal{G}$ always 
traverse $p$. Equivalently, the paths from $S_i$ to $a$ and to $d$ in $G(R)$ also always 
traverse $p$ for all $i\in[m]$. Since $p\in \alpha_1$, the edge weights of edges $(S_i,b)$ in $G_{\alpha_1}$ remain unchanged for all $i\in [m]$. Thus, for any general graph, only the edges $(b,b)$ and $(p,b)$ can be modified in the corresponding reduced $G_{\alpha_1}$.

Now, for any digraph $\mathcal{G}$, either of the  following conditions can hold:
\begin{enumerate}
    \item[C1.] 
    each path (that exists in $\mathcal{G}$) from $b$ to $a$ and $b$ to $d$  traverses $p$, in other words, $b\notin \mathcal{N}_p$,
    \item[C2.] a path from $b$ to $a$ or to $d$ exists that does not traverse $p$, \textit{i.e.,} $b\in \mathcal{N}_p$.
\end{enumerate}
 
If C1) holds, then each path from $b$ to $b$ in $G(R)$ that traverses the edges $(a,b)$ or $(d,b)$ must traverse $p$. Hence, the edge weight of $(b,b)$ remains unchanged in $G_{\alpha_1}$, leaving only edge $(p,b)$, which can get modified. By Lemma \ref{lm:Basic_properties}, it follows that $R/\alpha_1^c$ is a loopless Laplacian matrix with row-sums equal to $0$. Thus, each node in $G_{\alpha_1}$ has zero in-degree.
Hence, the weight of a single edge $(p,b)$ cannot change in  $G_{\alpha_1}$ and the influence centrality vector remains constant.

If C2) holds and a path from $b$ to $a$ or $d$ (or both) exists that does not traverse $p$. Hence, the edge-weights of both the edges $(p,b)$ and $(b,b)$ in $G_{\alpha_1}$ can get modified. 
In this scenario, we further reduce $G_{\alpha_1}$.

\textbf{Step 2:} We reduce $G_{\alpha_1}$ to $G_{\alpha_2}$ by considering $\alpha_2=\{p,S_1,...,S_m,O\}$ and $\alpha_2^c=\{b\}$.
Here, for the modified edges $(b,b)$ and $(p,b)$ in $G_{\alpha_1}$, the conditions (ii) and (iv) of Theorem \ref{lemma:modification_propagates} hold simultaneously. Thus, we examine the Schur complement $(R/\alpha_1^c)/\alpha_2^c$ to determine the change in edge weights of the reduced $G_{\alpha_2}$.  

Let $R_1=(R/\alpha_1^c)$ with the rows (and columns) indexed according to set $\alpha_1$.  Note that under C2), each path from $S_i$ to $b$ in $G(R)$ passes through $p$ for each $i \in [m]$; otherwise $a,d\notin \mathcal{N}_p$. Consequently, $R_1[\alpha_2^c,\alpha_2]=[h_1 ~ \mathbf{0}_{m+1}]$ where $h_1\in \mathbb{R}$ is a scalar. Additionally, in $G_{\alpha_1}$, node $b$ can have a path only to $p$ and $O$, resulting in $R_1[\alpha_2,\alpha_2^c]={\begin{bmatrix}
    h_2 & \mathbf{0}_{m} & h_3
\end{bmatrix}^T}$ where $h_2,h_3 \in \mathbb{R}$.

Under the edge modifications, the edges $(p,b)$ and  $(b,b)$ in $G_{\alpha_1}$ are affected. Thus, the parameter $h_1$ is modified to $\hat{h}_1$ and $R_1[\alpha_2^c]$ is modified to $\hat{R}_1[\alpha_2^c]$, yielding the matrix $\hat{R}_1$. To identify the affected edges in $G_{\alpha_2}$, we determine
 $\Delta R_2=\hat{R}_1/\alpha_2^c-{R}_1/\alpha_2^c={\begin{bmatrix}
    (\hat{h}_1\hat{R}_2[\alpha_3^c])^{-1}-h_1{R}_2[\alpha_3^c])^{-1})h_2 & \mathbf{0}_{m+1} \\
    0_{m} & \mathbf{0}_{m,m+1}\\
    (\hat{h}_1\hat{R}_2[\alpha_3^c]^{-1}-h_1{R}_2[\alpha_3^c]^{-1})h_3 & \mathbf{0}_{m+1} 
\end{bmatrix}}$. 
Similar to case C1), a single incoming edge $(p,p)$ at $p$ and $(p,O)$ at $O$ is modified in $G_{\alpha_2}$. By Statement 2) in Lemma \ref{lm:basic_properties_1}, the in-degree of the reduced graphs is always zero. Hence, the edge weight of a single incoming edge cannot change and  
the influence centrality vector remains unchanged. \hfill$\blacksquare$ 

\textbf{\textit{Proof of Theorem} \ref{thm:useful_mod}:}
To prove that $(a,b,d)$ increases $c_s$, we iteratively reduce $G(R)$ and determine the change in edge weights of the reduced graphs. 

\textbf{Step 1:} First, we consider $\alpha_1=(a,b,s,S_1,...,S_m,O)$ and reduce $G(R)$ to $G_{\alpha_1}$.  
The following edges in $G_{\alpha_1}$ are impacted:
\begin{itemize}
 \item Since $d\in \alpha_1^c$ and $b\in \alpha_1$, by Theorem \ref{lemma:modification_propagates}, the edge weight of $(j,b)$ in $G_{\alpha_1}$ increases if a path $j$ to $b$ in $G(R)$ contains $(d,b)$ and consists only of nodes from $\alpha_1^c$ (except $j$ and $b$).
 
    \item As nodes $a,b \in \alpha_1$, the edge $(a,b)$ in $G(R)$ only affects the edge weight of corresponding edge $(a,b)$ in $G_{\alpha_1}$. Notably, edge $(a,b)$ in $G_{\alpha_1}$ can also be affected by $(d,b)$ (in $G(R)$) if there exists a suitable path in $G(R)$ from $a$ to $b$ that traverses $(d,b)$. Thus, the edge weight of $(a,b)$ may increase or decrease.
\end{itemize}

The exact change in the edge weights of the incoming edges at $b$ can be determined as:
\begingroup
\setlength{\abovedisplayskip}{2pt}
\setlength{\belowdisplayskip}{2pt}
\begin{align*}
    \Delta R_1&=\hat{R}/\alpha_1^c-R/\alpha_1^c\\
    &=\hat{R}[\alpha_1]-{R}[\alpha_1]+(R[\alpha_1,\alpha_1^c]-\hat{R}[\alpha_1,\alpha_1^c]){R}[\alpha_1^c]^{-1}R[\alpha_1^c,\alpha_1]
\end{align*}
\endgroup
The second equality follows because only the $(b,a)^{th}$ and $(b,d)^{th}$ entries differ in $R$ and $\hat{R}$. 

Let the nodes in $\alpha_1^c$ be arranged such that $d$ is the first node. In this scenario,  we focus on $e_b^T \Delta R_1$ as only the entries of $b^{th}$ row of $R_1$ change, 
\begingroup
\setlength{\abovedisplayskip}{2pt}
\setlength{\belowdisplayskip}{2pt}
\begin{align*}
   e_b^T \Delta R_1&= \begin{bmatrix}
        -y & \mathbf{0}_{m+3} 
   \end{bmatrix}^T +
   &\begin{bmatrix}
        -y & \mathbf{0}_{|\alpha_1^c|-1} 
   \end{bmatrix}^T{R}[\alpha_1^c]^{-1}R[\alpha_1^c,\alpha_1]
\end{align*}
\endgroup
where $y=(1-\beta_b)w$.
By Lemma \ref{lm:submatrix_M}, the submatrix ${R}[\alpha_1^c]^{-1}$ is a non-singular $M$ matrix and its inverse ${R}[\alpha_1^c]^{-1}$ has non-negative entries. Additionally, the off-diagonal elements of $R$ are non-positive. Thus,
\begingroup
\setlength{\abovedisplayskip}{2pt}
\setlength{\belowdisplayskip}{2pt}
\begin{align}
\label{eqn:change_R1}
   e_b^T \Delta R_1
   &= \big[ -y \quad \mathbf{0}_{m+3} \big] \nonumber \\
     & + \big[ \delta_{a}y \quad \delta_{b}y \quad \delta_{s}y \quad
      \delta_{S_{1}}y \quad \cdots \quad \delta_{S_m}y \quad 0 \big].
\end{align}
\endgroup
where $\delta_{j}\geq 0$ for each $j\in \alpha_1$.
The following changes occur in $G_{\alpha_1}$:
\begin{itemize}
    \item the edge weights of $(j,b)$ for $j\in \alpha_1\setminus \{a,S_s,O\}$ increases by $y\delta_{j}$. 
    \item the edge weight of $(a,b)$ is changed by $y(-1+\delta_{a})$.
        \item Node $S_s$ has only one outgoing edge $(S_s,s)$ in $G(R)$. Since $s\in \alpha_1$, the edge $(S_s,b)$ does not exist in $G_{\alpha_1}$ both before and after modification. Similarly, by construction, the edge $(O,b)$ does not exist. Thus, the corresponding entries remain zero in $e_b^T \Delta R_1$.
\end{itemize}
Since $R/\alpha_1^c$ and $\hat{R}/\alpha_1^c$ are loopless Laplacian matrices, $e_b^T\Delta R_1 \mathbf{1}_{m+2}=0$. Consequently, $\sum_{j\in \alpha_1\setminus \{S_s,O\}} \delta_{k}=1$ and $\delta_j \in [0,1]$ for all $j \in \alpha_1$. 
Since only the edge weight of $(a,b)$ in $G_{\alpha_1}$ can decrease, it follows that unless there is no change in the $G_{\alpha_1}$, the edge weight of $(a,b)$ always decreases.

\textbf{Step 2:} Now, we further reduce $G_{\alpha_1}$ to $G_{\alpha_2}$ with $\alpha_2=\{s,S_1, S_2,...,S_m,O\}$ and $\alpha_2^c=\{a,b\}$. Since $\alpha_2 \supseteq\{S_1, S_2,...,S_m,O\}$, the Schur complement $R_1/\alpha_2^c$ is well-defined. We examine the changes in $G_{\alpha_2}$ due to edge modifications in $G_{\alpha_1}$. 

Recall that only the incoming edges at $b$ in $G_{\alpha_1}$ are modified. Hence, only the entries of $b^{th}$ row in $R_1$ are affected. Since $b\in \alpha_2^c$, we determine the change in edges in $G_{\alpha_2}$ using the Schur complement. 
Let $R_2=R_1/\alpha_2^c$, the change in $R_2$ is given as $\Delta R_2=\hat{R}_1/\alpha_2^c-{R}_1/\alpha_2^c$ is equal to,  \begin{align}
\label{eqn:change_in_R3}
    \tiny \Delta R_2=&R_1[\alpha_2,\alpha_2^c]\big(R_1[\alpha_2^c]^{-1}R_1[\alpha_2^c,\alpha_2]-\hat{R}_1[\alpha_2^c]^{-1}\hat{R}_1[\alpha_2^c,\alpha_2]\big) \nonumber \\
    =&{R_1[\alpha_2,\alpha_2^c] R_1[\alpha_2^c]^{-1}(R_1[\alpha_2^c,\alpha_2]-\hat{R}_1[\alpha_2^c,\alpha_2])} + \nonumber \\ &{R_1[\alpha_2,\alpha_2^c]({R}_1[\alpha_2^c]^{-1}- \hat{R}_1[\alpha_2^c])^{-1}\hat{R}_1[\alpha_2^c,\alpha_2]}
\end{align}
Next, we use the topological properties of $G_{\alpha_1}$ to determine the change in $\Delta R_2$. 
\begin{enumerate}
    \item  Only nodes $s$ and $O$ in set $\alpha_2$ have incoming edges in $G_{\alpha_1}$. Thus, $R_1[\alpha_2,\alpha_2^c]=\scriptsize{\begin{bmatrix}
        \times & \times \\
        \mathbf{0}_m &  \mathbf{0}_m \\
        \times & \times
    \end{bmatrix}}$ where $\times$ highlights the possible non-zero entries.
    \item Since $a$ is an endorser of $s$ and $s\in \alpha_1$, the edge $(S_j,a)$ does not exist in $G_{\alpha_1}$ for all $j\in [m]$. Hence, $R_1[\alpha_2^c,\alpha_2]=\scriptsize{\begin{bmatrix}
        \times & \mathbf{0}_{m+1} \\
        \times & \times
    \end{bmatrix}}$. Note that this property holds for $\hat{R}_1[\alpha_2^c,\alpha_2]$ as well. 
    \item Since only the incoming edges at node $b$ in $G_{\alpha_1}$ are modified, by eqn. \eqref{eqn:change_R1}, $R_1[\alpha_2^c,\alpha_2]-\hat{R}_1[\alpha_2^c,\alpha_2]$ equals,
\begin{align*}
\scriptsize{\begin{bmatrix}
    0& 0 & ... & 0 & 0& 0 & ... & 0 & 0 \\
    -y\delta_{s} &-y\delta_{S_1} & ... & -y\delta_{S_{s-1}} & 0& -y\delta_{S_{s+1}}& ... & -y\delta_{S_m} & 0
\end{bmatrix}   } 
\end{align*}
Note that since the edges $(S_s,b)$ and $(O,b)$ do not exist in $G_{\alpha_1}$, the corresponding entries are zero.
\item Let $R_1[\alpha_2^c]=\scriptsize{\begin{bmatrix}
    r_{aa}^1 & r_{ab}^1 \\
    r_{ba}^1 & r_{bb}^1
\end{bmatrix}}$
where $r_{ef}^1$ denotes the entry $(e,f)^{th}$ entry of $R_1$. 

Node $d$ is a non-endorser of $s$ and edge $(d,b)$ exists in $\mathcal{G}$, hence, node $b$ is also a non-endorser of $s$. Thus, each path in $G(R)$ from $b$ to $a$ passes through $s$; otherwise, $a$ would also become a non-endorser. Since $s\in \alpha_1$, the edge $(b,a)$ does not exist in $G_{\alpha_1}$. This means that $r_{ab}^1=0$. By eqn. \eqref{eqn:change_R1}, it follows that due to edge modifications in $G(R)$, $\hat{R}_1[\alpha_2^c]=\scriptsize{\begin{bmatrix}
    r_{aa}^1 & 0 \\
    r_{ba}^1-(1-\delta_a)y & r_{bb}^1+\delta_b y
\end{bmatrix}}$. Finally, simple calculations show that $R_1[\alpha_2^c]^{-1}-\hat{R}_1[\alpha_2^c]^{-1}=\scriptsize{\begin{bmatrix}
    0 & 0 \\
    k_1 & k_2
\end{bmatrix}}$ where $k_1 \in \mathbb{R},k_2 \in \mathbb{R}_{>0}$.
\end{enumerate}

Additionally, recall that each entry of $R_1[\alpha_2^c]^{-1}$ is positive and each entry of matrices $R_1[\alpha_2,\alpha_2^c]$ and $\hat{R}_1[\alpha_2,\alpha_2^c]$ is non-positive. Thus, by points 1-4, and eqn. \eqref{eqn:change_in_R3}, the following holds:
\begin{itemize}
    \item The edge weight of $(i,j)$ in $\mathcal{G}_{\alpha_2}$ increases for $i\in\alpha_2\setminus \{S_s,O,s\}$ and $j \in \{s,O\}$ 
    \item  The edge weight of $(S_s,j)$ in $\mathcal{G}_{\alpha_2}$ remains unchanged  and $j \in \{s,O\}$.
    \item Since $\hat{R}_1/\alpha_2^c$ is also a Laplacian matrix with row-sums $0$, the edge weight of $(s,j)$ in $\mathcal{G}_{\alpha_2}$ decreases for $j \in \{s,O\}$.
\end{itemize}
These modified edges result in the matrix $\hat{R}_2$.

\textbf{Step 3:} In this step, we reduce $G_{\alpha_2}$ to $G_{\alpha_3}$ with $\alpha_3=\{S_1,...,S_m,O\}$ and $\alpha_3^c=\{s\}$. As demostrated in Sec. \ref{subsec:KRIC}, the edge weight of edge $(S_j,O)$ in $G_{\alpha_3}$ is $-c_j$ for all $j\in [m]$. Thus, we evaluate the change in $c_s$.

From eqn. \eqref{eqn:Schur_complement}, it follows that $
 \hat{c}_j-{c}_j={R}_2[O,S_j]-\hat{R}_2[O,S_j] +\hat{R}_2[O,s]\hat{R}_2[s]^{-1}\hat{R}_2[s,S_j]- R_2[O,s]R_2[s]^{-1}{R}_2[s,S_j]$ for any $j\in [m]$. Note that the following properties hold in  $G_{\alpha_2}$: (i) the edge $(S_s,s)$ remains unchanged and (ii) the edge $(S_s,O)$ does not exist. Thus, we can simplify $\hat{c}_s-{c}_s=(\hat{R}_2[O,s]\hat{R}_2[s]^{-1}-R_2[O,s]R_2[s]^{-1}){R}_2[s,S_s]$. On further simplifying we get,
 $\hat{c}_s-{c}_s=R_2[O,s](R_2[s]^{-1}-\hat{R}_2[s]^{-1}){R}_2[s,S_s]+({R}_2[O,s]-\hat{R}_2[O,s])\hat{R}_2[s]^{-1}{R}_2[s,S_s]$. 

From Step 2, we know that the weight of edges $(s,s)$ and $(s,O)$ in $G_{\alpha_2}$ have reduced. Thus, $(\hat{R}_2[s]^{-1}-R_2[s]^{-1})>0$ and $\hat{R}_2[O,s]-{R}_2[O,s]<0.$ Additionally, we know that from Lemma \ref{lm:Q_is_M_matrix}, we know that both $R_2$ and $\hat{R}_2$ are $M$ matrices. Thus, the diagonal entry ($\hat{R}_2[s]>0$) is positive and the off-diagonal entries ($R_2[O,s]$ and $R_2[s,s]$) are negative. Therefore, $\hat{c}_j>c_j$ and the influence centrality of stubborn agent $s$ increases. \hfill$\blacksquare$

\section*{REFERENCES}
\vspace{-20pt}
\bibliographystyle{IEEEtran}
\bibliography{IEEEabrv,references_main}

\begin{thebibliography}{10}
\providecommand{\url}[1]{#1}
\csname url@samestyle\endcsname
\providecommand{\newblock}{\relax}
\providecommand{\bibinfo}[2]{#2}
\providecommand{\BIBentrySTDinterwordspacing}{\spaceskip=0pt\relax}
\providecommand{\BIBentryALTinterwordstretchfactor}{4}
\providecommand{\BIBentryALTinterwordspacing}{\spaceskip=\fontdimen2\font plus
\BIBentryALTinterwordstretchfactor\fontdimen3\font minus
  \fontdimen4\font\relax}
\providecommand{\BIBforeignlanguage}[2]{{%
\expandafter\ifx\csname l@#1\endcsname\relax
\typeout{** WARNING: IEEEtran.bst: No hyphenation pattern has been}%
\typeout{** loaded for the language `#1'. Using the pattern for}%
\typeout{** the default language instead.}%
\else
\language=\csname l@#1\endcsname
\fi
#2}}
\providecommand{\BIBdecl}{\relax}
\BIBdecl

\bibitem{cartwright1959studies}
D.~Cartwright \emph{et~al.}, \emph{Studies in social power}.\hskip 1em plus
  0.5em minus 0.4em\relax Publications of the Institute for Social Research:
  Research Center for Group Dynamics Series, Research Center for Group
  Dynamics, Institute for Social Research, University of Michigan, 1959.

\bibitem{friedkin1991centrality}
N.~E. Friedkin, ``Theoretical foundations for centrality measures,''
  \emph{American Journal of Sociology}, vol.~96, no.~6, pp. 1478--1504, 1991.

\bibitem{degroot1974consensus}
M.~H. DeGroot, ``Reaching a consensus,'' \emph{Journal of the American
  Statistical Association}, vol.~69, no. 345, pp. 118--121, 1974.

\bibitem{anderson1981foundations}
N.~H. Anderson, ``Foundations of information integration theory,'' 1981.

\bibitem{Community_Cleavage}
N.~E. Friedkin, ``The problem of social control and coordination of complex
  systems in sociology: A look at the community cleavage problem,'' \emph{IEEE
  Control Systems Magazine}, vol.~35, pp. 40--51, 2015.

\bibitem{10.1145/3511808.3557304}
W.~Xu, L.~Zhu, J.~Guan, Z.~Zhang, and Z.~Zhang, ``Effects of stubbornness on
  opinion dynamics,'' in \emph{Proceedings of the 31st ACM International
  Conference on Information \& Knowledge Management}, ser. CIKM '22, p.
  2321–2330.

\bibitem{Lingfei_wang}
L.~Wang, G.~Chen, C.~Bernardo, Y.~Hong, G.~Shi, and C.~Altafini, ``Social power
  games in concatenated opinion dynamics,'' \emph{IEEE Transactions on
  Automatic Control}, vol.~69, no.~11, pp. 7614--7629, 2024.

\bibitem{L_wang_parallel}
L.~Wang, Y.~Xing, S.~Huang, C.~Altafini, and K.~H. Johansson, ``Social power
  games for parallel {F}riedkin–{J}ohnsen models,'' \emph{IEEE Transactions
  on Automatic Control}, vol.~71, no.~5, pp. 3074--3089, 2026.

\bibitem{feed_algo}
L.~Edelson, F.~Haugen, and D.~McCoy, ``A comparative survey of algorithmic feed
  recommendation system designs,'' \emph{ACM Trans. Recomm. Syst.}, 2025.

\bibitem{ancona2022model}
C.~Ancona, F.~L. Iudice, F.~Garofalo, and P.~De~Lellis, ``A model-based opinion
  dynamics approach to tackle vaccine hesitancy,'' \emph{Scientific Reports},
  vol.~12, no.~1, p. 11835, 2022.

\bibitem{wang2025addinglinks}
L.~Wang, Y.~Xing, Y.~Yi, M.~Cao, and K.~H. Johansson, ``Adding links wisely:
  how an influencer seeks for leadership in opinion dynamics?'' {J}une 2025,
  arXiv:2506.12463.

\bibitem{gt_attract}
Y.~Ao and Y.~Jia, ``Agents attraction competition in an extended
  {F}riedkin-{J}ohnsen social network,'' \emph{IEEE Transactions on Control of
  Network Systems}, vol.~10, no.~3, pp. 1100--1112, 2023.

\bibitem{ZHU2025115090}
L.~Zhu and Z.~Zhang, ``Opinion maximization in social networks via link
  recommendation,'' \emph{Theoretical Computer Science}, vol. 1033, p. 115090,
  2025.

\bibitem{content_filtering}
T.~Bansal, M.~Das, and C.~Bhattacharyya, ``Content driven user profiling for
  comment-worthy recommendations of news and blog articles,'' in
  \emph{Proceedings of the 9th ACM Conference on Recommender Systems}, ser.
  RecSys '15, p. 195–202.

\bibitem{collab_filtering}
M.~Eirinaki, M.~D. Louta, and I.~Varlamis, ``A trust-aware system for
  personalized user recommendations in social networks,'' \emph{IEEE
  Transactions on Systems, Man, and Cybernetics: Systems}, vol.~44, no.~4, pp.
  409--421, 2014.

\bibitem{Survey_recommendations}
R.~Chen, Q.~Hua, Y.-S. Chang, B.~Wang, L.~Zhang, and X.~Kong, ``A survey of
  collaborative filtering-based recommender systems: From traditional methods
  to hybrid methods based on social networks,'' \emph{IEEE Access}, vol.~6, pp.
  64\,301--64\,320, 2018.

\bibitem{etienne2024mimetic}
H.~Etienne and F.~Charton, ``A mimetic approach to social influence on
  instagram,'' \emph{Philosophy \& Technology}, vol.~37, no.~2, p.~65, 2024.

\bibitem{shrinate2025opinionclusteringfriedkinjohnsenmodel}
A.~Shrinate and T.~Tripathy, ``Opinion clustering under the
  {F}riedkin-{J}ohnsen model: Agreement in disagreement,''
  {https://arxiv.org/abs/2509.11045v3}.

\bibitem{Kron_red_digraphs}
T.~Sugiyama and K.~Sato, ``Kron reduction and effective resistance of directed
  graphs,'' \emph{SIAM Journal on Matrix Analysis and Applications}, vol.~44,
  no.~1, pp. 270--292, 2023.

\bibitem{dorfler2012kron}
F.~Dorfler and F.~Bullo, ``Kron reduction of graphs with applications to
  electrical networks,'' \emph{IEEE Transactions on Circuits and Systems I:
  Regular Papers}, vol.~60, no.~1, pp. 150--163, 2012.

\bibitem{Plemmons}
A.~Berman and R.~J. Plemmons, ``{M}-matrices,'' in \emph{Nonnegative Matrices
  in the Mathematical Sciences}.\hskip 1em plus 0.5em minus 0.4em\relax
  Philadelphia, PA, USA: Society for Industrial and Applied Mathematics (SIAM),
  1994, ch.~6, pp. 132--164.

\bibitem{JOHNSON202484}
C.~Johnson, C.~Marijuán, M.~Pisonero, and I.~Spitkovsky, ``Diagonal entries of
  inverses of diagonally dominant matrices,'' \emph{Linear Algebra and its
  Applications}, vol. 692, pp. 84--90, 2024.

\bibitem{CHEBOTAREV2002253}
P.~Chebotarev and R.~Agaev, ``Forest matrices around the {L}aplacian matrix,''
  \emph{Linear Algebra and its Applications}, vol. 356, no.~1, pp. 253--274,
  2002.

\bibitem{bullo}
F.~Bullo, \emph{Lectures on Network Systems}, {1.7}~ed.\hskip 1em plus 0.5em
  minus 0.4em\relax Kindle Direct Publishing, 2024.

\bibitem{friedkin1990opinions}
N.~E. Friedkin and E.~C. Johnsen, ``Social influence and opinions,'' \emph{The
  Journal of Mathematical Sociology}, vol.~15, pp. 193--206, 1990.

\bibitem{victor}
V.~Amelkin and A.~K. Singh, ``Fighting opinion control in social networks via
  link recommendation,'' in \emph{Proceedings of the 25th ACM SIGKDD
  International Conference on Knowledge Discovery \& Data Mining}, ser. KDD
  '19, p. 677–685.

\bibitem{sherman1950}
J.~Sherman and W.~J. Morrison, ``Adjustment of an inverse matrix corresponding
  to a change in one element of a given matrix,'' \emph{The Annals of
  Mathematical Statistics}, vol.~21, no.~1, pp. 124--127, 1950.

\bibitem{gionis2013opinion}
A.~Gionis, E.~Terzi, and P.~Tsaparas, ``Opinion maximization in social
  networks,'' in \emph{Proceedings of the 2013 SIAM international conference on
  data mining}, pp. 387--395.

\bibitem{fan_pages}
L.~de~Vries, S.~Gensler, and P.~S. Leeflang, ``Popularity of brand posts on
  brand fan pages: An investigation of the effects of social media marketing,''
  \emph{Journal of Interactive Marketing}, vol.~26, no.~2, pp. 83--91, 2012.

\bibitem{harm_ful_content}
C.~Coupette, S.~Neumann, and A.~Gionis, ``Reducing exposure to harmful content
  via graph rewiring,'' in \emph{Proceedings of the 29th ACM SIGKDD Conference
  on Knowledge Discovery and Data Mining}, ser. KDD '23, p. 323–334.

\bibitem{graham2008note}
A.~J. Graham and D.~A. Pike, ``A note on thresholds and connectivity in random
  directed graphs,'' \emph{Atl. Electron. J. Math}, vol.~3, no.~1, pp. 1--5,
  2008.

\end{thebibliography}

\vskip -2\baselineskip plus -1fil
\begin{IEEEbiography}[{\includegraphics[width=0.9in,height=1in]{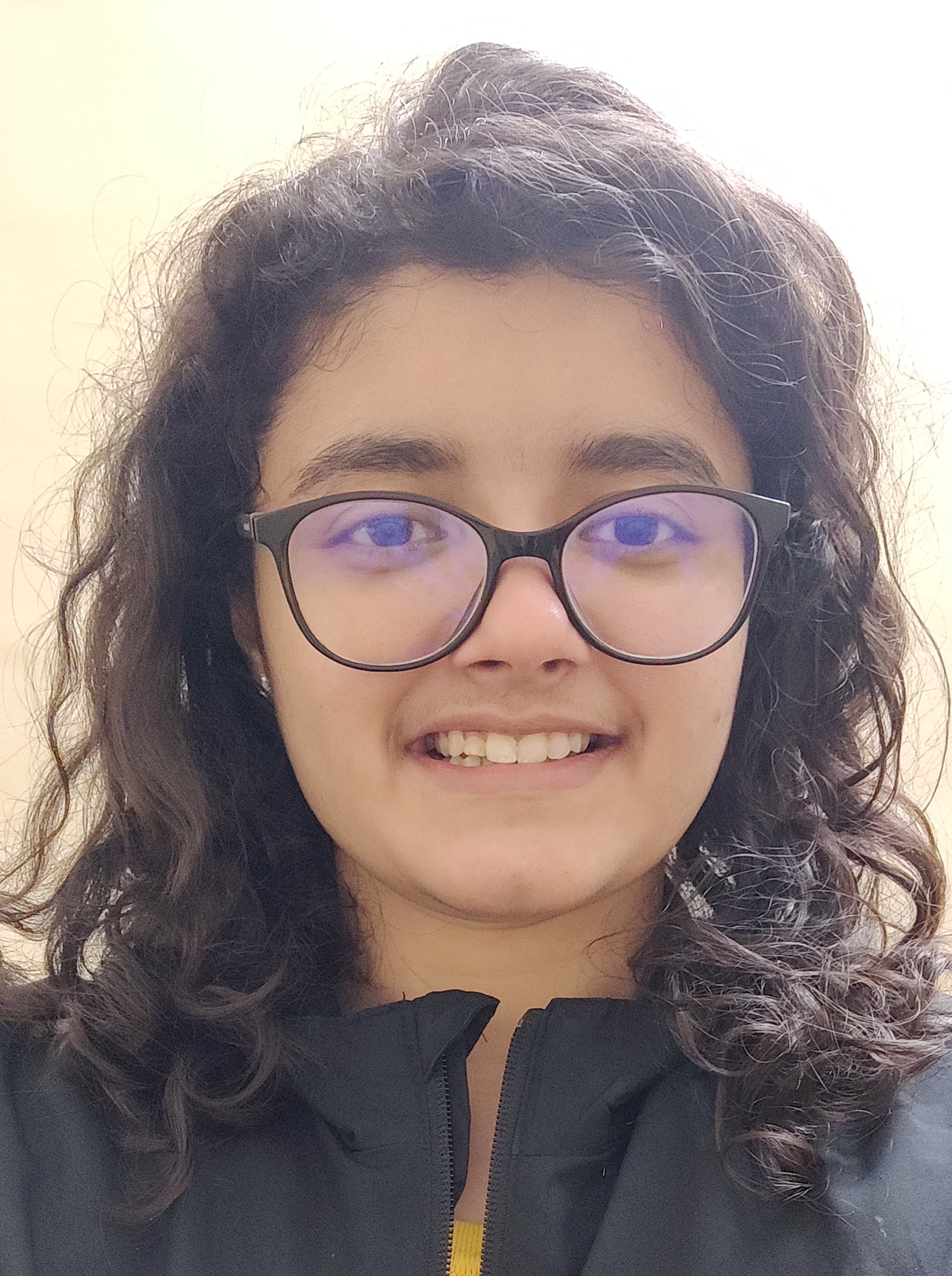}}]{Aashi Shrinate} received her B. Tech degree in Electrical Engineering from Motilal Nehru National Institute of Technology, Allahabad, India in 2020. She is working towards a PhD degree in control and automation specialization in Distributed Control and Decision Lab, Department of Electrical Engineering at the Indian Institute of Technology, Kanpur. She has been receiving the Prime Minister Research Fellowship from 2023. Her research focuses on networked dynamical systems with applications to opinion dynamics in social networks and robotic networks. 
\end{IEEEbiography}
\vskip -2\baselineskip plus -1fil
\begin{IEEEbiography}[{\includegraphics[width=0.9in,height=1.05in,clip]{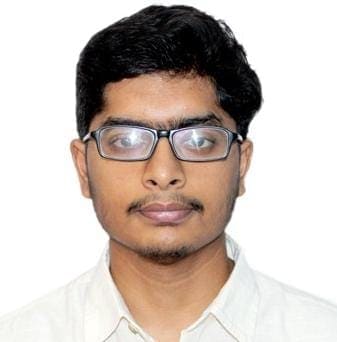}}]{Aravind Seshadri} received his B. Tech degree in Electrical Engineering
from Indian Institute of Technology Kanpur, Uttar Pradesh, India in 2025.
He currently works as a Machine Learning Engineer at Adobe Systems. His
research interests include Computer Vision, Robotics, and opinion dynamics
in social networks.
\end{IEEEbiography}
\vskip -2\baselineskip plus -1fil
\begin{IEEEbiography}[{\includegraphics[width=1in,height=1.05in,clip]{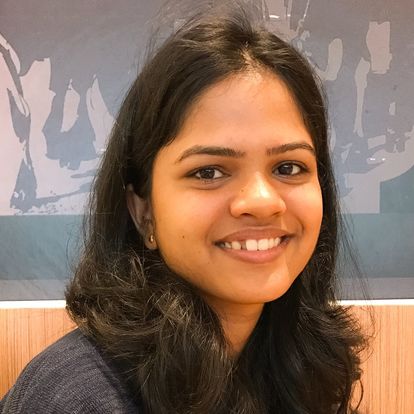}}]{Twinkle Tripathy} is currently an Assistant Professor in the Department of Electrical Engineering of IIT Kanpur. She received a Dual Degree of MTech. and Ph.D. at Systems \& Control Engineering, IIT Bombay in Dec. 2016. She started her post-doctoral tenure at the School of Electrical \& Electronic Engineering, NTU, Singapore. After serving there for a year, she joined the Faculty of Aerospace Engineering, Technion – Israel Institute of Technology as a post-doctoral fellow. Her research interests broadly include control and guidance of autonomous systems, cyclic pursuit strategies and opinion dynamics.
\end{IEEEbiography}
\vskip -2\baselineskip plus -1fil
\begin{IEEEbiography}[{\includegraphics[width=1in,height=1.05in,clip]{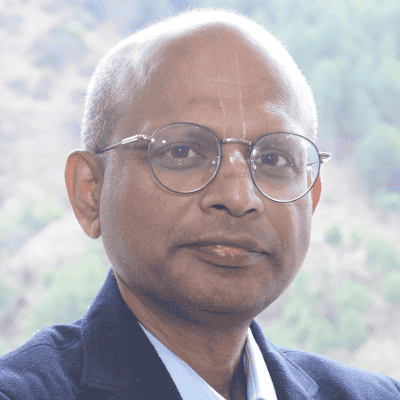}}]{Laxmidhar Behera}
received the B.Sc.(Engg.) and M.Sc.(Engg.) degrees
from the National Institute of Technology (NIT),
Rourkela, India, in 1988 and 1990, respectively, and
the Ph.D. degree in electrical engineering from IIT
Delhi, New Delhi, India, in 1996.
He is currently the Director of IIT Mandi, Mandi,
India, and a Senior Professor with the Department
of Electrical Engineering, IIT Kanpur, Kanpur, India.
He pursued post-doctoral research with the German
National Research Centre for Information Technology, Sankt Augustin, Germany. He has over 28 years of teaching and research
experience in intelligent systems and control, robotics, warehouse automation,
brain–computer interfaces, drone technologies, and mental healthcare. Dr. Behera is a Fellow of INAE, a Distinguished Alumnus of NIT Rourkela,
and a recipient of the National Systems Gold Medal (2023) from the Systems
Society of India. His team achieved top positions in the Amazon Robotics
Challenge 2017. He is an Associate Editor of IEEE Transactions on Systems, Man, and Cybernetics:Systems.
\end{IEEEbiography}
\begin{IEEEbiography}[{\includegraphics[width=1in,height=1.05in,clip]{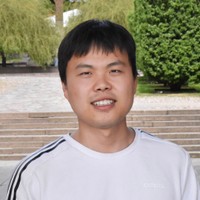}}]{Lingfei Wang}
    received the B.Sc. degree in applied mathematics from Nankai University, Tianjin, China, in 2017, and the Ph.D. degree in
systems theory from Academy of Mathematics
and Systems Science, Chinese Academy of Sciences, Beijing, China, in 2023.
In 2020 and 2022, he was a Visiting Ph.D.
Student with the Division of Automatic Control,
Linkoping University, Linkoping, Sweden. Since
March 2023, he has been a Postdoctoral Researcher with the Division of Decision and Control Systems, KTH Royal Institute of Technology, Stockholm, Sweden.
His research interests include multi-agent systems, social opinion dynamics, and game theory
\end{IEEEbiography}
\begin{IEEEbiography}[{\includegraphics[width=1in,height=1.05in,clip]{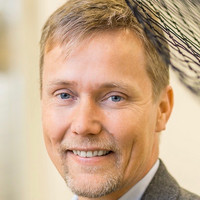}}]{Karl H. Johansson} is Swedish Research Council
Distinguished Professor in Electrical Engineering
and Computer Science at KTH Royal Institute of
Technology in Sweden and Founding Director of
Digital Futures. He earned his MSc degree in Electrical Engineering and PhD in Automatic Control
from Lund University. He has held visiting positions at UC Berkeley, Caltech, NTU and other
prestigious institutions. His research interests focus
on networked control systems and cyber-physical
systems with applications in transportation, energy,
and automation networks. For his scientific contributions, he has received
numerous best paper awards and various distinctions from IEEE, IFAC,
and other organizations. He has been awarded Distinguished Professor by
the Swedish Research Council, Wallenberg Scholar by the Knut and Alice
Wallenberg Foundation, Future Research Leader by the Swedish Foundation
for Strategic Research. He has also received the triennial IFAC Young Author
Prize and IEEE CSS Distinguished Lecturer. He is the recipient of the
2024 IEEE CSS Henrik W. Bode Lecture Prize. His extensive service to
the academic community includes being President of the European Control
Association, IEEE CSS Vice President Diversity, Outreach \& Development,
and Member of IEEE CSS Board of Governors and IFAC Council. He has
served on the editorial boards of Automatica, IEEE TAC, IEEE TCNS and
many other journals. He has also been a member of the Swedish Scientific
Council for Natural Sciences and Engineering Sciences. He is Fellow of both
the IEEE and the Royal Swedish Academy of Engineering Sciences
    
\end{IEEEbiography}

\end{document}